\documentclass[a4paper,UKenglish,cleveref, autoref, thm-restate]{lipics-v2021}

\title{Solving Stackelberg Vertex Cover on trees using split and join} 

\titlerunning{Stackelberg Vertex Cover on trees} 

\author{Dominik Scheder}{TU Chemnitz, Germany}{dominik.scheder@informatik.tu-chemnitz.de}{https://orcid.org/0000-0002-9360-7957}{}

\author{Johannes Tantow}{TU Chemnitz, Germany}{johannes.tantow@informatik.tu-chemnitz.de}{https://orcid.org/0009-0006-0408-6966}{}

\authorrunning{D. Scheder and J. Tantow} 

\Copyright{Dominik Scheder and Johannes Tantow} 



\category{} 

\relatedversion{} 

\acknowledgements{We want to thank Lennart Kauther for helpful discussions}

\nolinenumbers 

\hideLIPIcs  

\EventEditors{John Q. Open and Joan R. Access}
\EventNoEds{2}
\EventLongTitle{42nd Conference on Very Important Topics (CVIT 2016)}
\EventShortTitle{CVIT 2016}
\EventAcronym{CVIT}
\EventYear{2016}
\EventDate{December 24--27, 2016}
\EventLocation{Little Whinging, United Kingdom}
\EventLogo{}
\SeriesVolume{42}
\ArticleNo{23}
\usepackage{graphicx} 
\usepackage{amsmath}
\usepackage{amssymb}
\usepackage{amsthm}
\usepackage{cleveref}
\usepackage[bold, full]{complexity}
\usepackage{subcaption}
\usepackage{tikz}
\usetikzlibrary{trees,positioning}
\usepackage{tcolorbox}
\usepackage{thmtools}
\usepackage{float}
\usepackage{pythonhighlight}
\usepackage{soul}
\inputencoding{utf8}
\usepackage{wrapfig}

\newtheorem{problem}[theorem]{Problem}

\newcommand{\opt}{\textnormal{opt}}

\renewcommand{\R}{\mathbb{R}}
\newcommand{\feasible}{\textnormal{feasible}}

\newcommand{\im}{\textnormal{Im}}
\newcommand{\N}{\mathbb{N}}
\newcommand{\rev}{\textnormal{Rev}}
\newcommand{\Rev}{\textnormal{Rev}}
\newcommand{\Children}{\textnormal{Children}}
\newcommand{\children}{\textnormal{children}}
\newcommand{\capacity}{\textnormal{cap}}
\newcommand{\pen}{\textnormal{pen}}
\newcommand{\vc}{\textnormal{vc}}
\newcommand{\vca}{\textnormal{vca}}
\newcommand{\nonneg}{\mathbb{R}_0^+}
\newcommand{\nonneginf}{\nonneg \cup \{\infty\}}
\newcommand{\bag}{B}

\newcommand{\bags}{\mathcal{B}}
\newcommand{\decomptree}{\mathcal{T}}
\newcommand{\decompedges}{\mathcal{F}}
\newcommand{\mathimg}[1]{\vcenter{\hbox{\includegraphics[height=1.2em]{#1}}}}

\newcommand{\inlineimg}[1]{\raisebox{\dimexpr 0.5ex - 0.5\height}{\includegraphics[height=1.2em]{#1}}}

\begin{document}

\maketitle

\begin{abstract}
The Stackelberg Vertex Cover problem is a bilevel optimization problem with two players on a graph $G = (F \cup P, E)$ where each vertex from $F$ has a weight and the first player selects a price for each vertex in $P$. Afterwards, the second player finds a minimum weight vertex cover $X$ and the first player receives the set price for each vertex from $X \cap P$. The goal is to maximize the revenue of the first player.

This problem was recently shown to be NP-complete for bipartite graphs while being solvable in linear 
time on paths. We present four new algorithms for solving Stackelberg Vertex Cover on certain kinds of 
graphs: (1) a pseudo-polynomial algorithm working on general trees when all weights are integer with a runtime linear in the number of vertices and cubic in the maximum weight
(2) a generalization of (1) for bipartite graphs with integer weights and a tree decomposition that is FPT in the maximum weight and the treewidth, (3) a strongly polynomial algorithm 
for rooted trees having the property that the least common ancestor of any two vertices from $P$ is again in $P$
(this case includes paths); 
and (4) an 
FPT-algorithm for trees, where the parameter is the maximum number $P$-vertices $v_i$ that an $F$-vertex $u$ can 
reach while using no other $P$-vertices. 

These algorithms are based on a lemma that allows us to split instances at a vertex $u$ 
into multiple sub-instances, 
which follows from LP duality and integrality of the vertex cover LP on bipartite graphs. 
The lemma requires that the minimum vertex covers of the sub-instances agree on $u$ (either all include $u$ 
or all don't). For this we introduce the concept of {\em commitments}.
We show that the Stackelberg Vertex Cover problem with commitments is weakly NP-complete. An open question is the non-bipartite case as there is an explicit counterexample showing that the split-and-join technique does not work.
\end{abstract}


\newpage

\section{Introduction}

A Stackelberg Game is a multi-player game with one distinguished player, the {\em leader}, and one or more {\em followers}. The game consists of a set of {\em items}, some of which are owned by the leader and 
are called {\em priceable}; the other items have some fixed price.
In a first phase the leader will select prices for the priceable items. In the second phase 
the followers will buy items to solve some goal. If an item  belonging to the leader is bought, she will receive 
the price she set for it. The leader wants to set prices such as to maximize her revenue.
This game is based on the economical model of von Stackelberg\cite{stackelberg}. In computer science, these games are often studied as a bilevel optimization problem where the followers need to solve a combinatorial optimization problem.

One of the first examples was the study of the Stackelberg shortest path problem where the leader owns a subset of the edges and sets prices which are interpreted as tolls. Afterwards, the follower needs to solve a shortest path problem and buys the edges on the path. Maximizing the leader's revenue  was shown to be \NP{}-complete for general graphs by Labb{\'e}, Marcotte and Savard~\cite{stackShortPath}. Further work on this problem was by Roch, Savard and Marcotte in~\cite{DBLP:journals/networks/RochSM05}, by Joret in~\cite{DBLP:journals/networks/Joret11} or by Briest, Chalermsook, Khanna,  Laekhanukit and Nanongkai~\cite{DBLP:conf/wine/BriestCKLN10}. Other combinatorial problems studied were minimum spanning trees by Cardinal et al.~\cite{DBLP:journals/algorithmica/CardinalDFJLNW11}, packing and scheduling problems by B{\"o}hnlein, Schaudt und Schauer~\cite{DBLP:conf/wads/BohnleinSS19} or shortest path trees by Bil{\`o}, Gual{\`a}, Proietti and Widmayer~\cite{DBLP:conf/wine/BiloGPW08}.

Recently, a general toolkit was shown that proves $\Sigma_2^P$-completeness of the Stackelberg game for most \NP{}-complete problems by Gr{\"u}ne, Henke, Rotenberg and Wulf~\cite{DBLP:conf/esa/GruneHRW26}. Thus, there is now mainly an interest in problems where the follower needs to solve a problem from \P. All problems from $P$ admit \FPT-algorithms for their Stackelberg game when parameterized by the number of priceable items as shown by B{\"o}hnlein, Kratsch and Schaudt~\cite{DBLP:conf/icalp/BohnleinKS17}.

\subsection{Problem description}

We focus on the vertex cover problem, which is in \P{} for bipartite graphs. 
Stackelberg Vertex Cover was introduced in~\cite{DBLP:journals/algorithmica/BriestHK12} by Briest, Hoefer and Krysta.
We are given a graph $G = (V, E)$ with $V = P\ \dot{\cup}\ F$
and a weight function $w : F \to \nonneginf$. Vertices in $P$ 
are called {\em priceable} and vertices in $F$ are called 
{\em fixed-price}.
There are two players, the leader and the follower. In the first phase the leader 
picks a pricing function $p : P \to \nonneginf$, so 
$(G, p \cup w)$ is just a usual graph with vertex weights. 
In the second phase the follower chooses a vertex cover $X$ of $G$ of minimum weight under $p \cup w$. 
The {\em revenue} for the leader for this vertex cover is 
$\rev(G, p, X) := \sum_{x \in P \cap X} p(x)$.
If multiple minimum weight vertex covers exists, the follower will 
pick a vertex cover that maximizes the revenue for the leader.\footnote{This is a standard convention in 
the Stackelberg literature. If the follower does not follow it, the leader can enforce it by reducing 
the price of each priceable item that should be sold by $\varepsilon$.}
The leader wants to find a pricing function that maximizes their 
revenue.

This problem is \NP-complete already for bipartite graphs, as shown by Jungnitsch, Peis and Schr{\"o}der~\cite{DBLP:journals/mor/JungnitschPS22}, even when all prices are integers bounded by a constant. It admits a 2-approximation~\cite{DBLP:journals/algorithmica/BriestHK12}. The special case that all priceable vertices lie on the same side of the bipartition is solvable in polynomial~time~\cite{DBLP:journals/algorithmica/BriestHK12}. Recently, Eickhoff, Kauther and Peis~\cite{DBLP:conf/sagt/EickhoffKP23} gave a linear-time algorithm for {\em paths}.  Although paths are among the simplest graphs imaginable, their algorithm is highly non-trivial. 

One of our main result shows that the problem is pseudo-polynomial for {\em trees}. This is in a way best possible: if weights are integer and given in binary, Stackelberg VC on trees is \NP-complete, as shown by Eickhoff, Kauther and Peis~\cite{NPCTree} using a reduction from Subset Sum. In case the set $P$ is placed particularly favorably in the tree, a notion we formalize below, we get a strongly polynomial algorithm; this case includes paths.

\subsection{Results and ideas}

\begin{figure}
    \begin{subfigure}[t]{.3\textwidth}
        \centering
        \includegraphics[scale=.4]{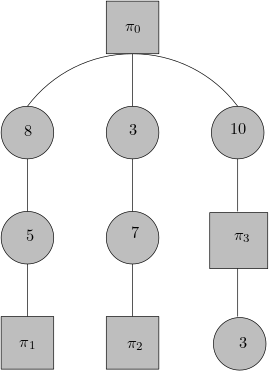}
        \caption{The tree before}
    \end{subfigure}
    \begin{subfigure}[t]{.3\textwidth}
        \centering
        \includegraphics[scale=.4]{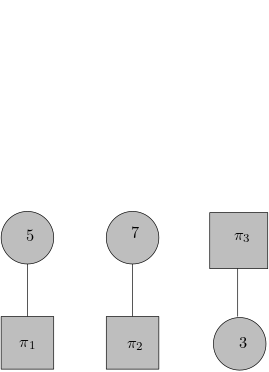}
        \caption{Splitting the tree after setting $\pi_0$ to $\infty$}
    \end{subfigure}
    \begin{subfigure}[t]{.3\textwidth}
        \centering
        \includegraphics[scale=.3]{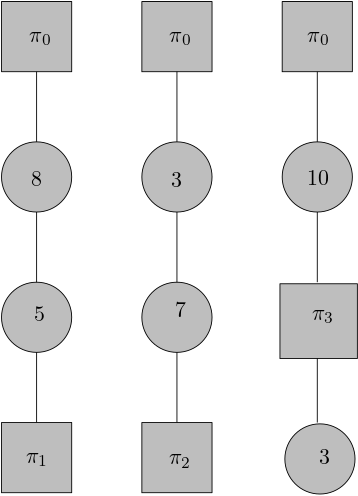}
        \caption{Splitting the instance when $\pi_0$ is supposed to be sold.}
        \label{subfigure-split-yes}
    \end{subfigure}
    \caption{Reducing an instance by splitting.}
    \label{example-tree}
\end{figure}

Consider the tree in \cref{example-tree} where priceable vertices are the rectangular vertices. If we (the leader)
were to set the price of $\pi_0$ to $\infty$, the follower would be forced to buy all its neighbors 
(lest they pay an infinite price); this splits the instances into individual subtrees with only one 
priceable vertex, and we can solve the rest easily achieving a revenue of 15.

Thus, if we decide to certainly {\em not} sell the priceable vertex $\pi_0$, we can split the instance 
by removing $\pi_0$ and all its neighbors. 
On the other hand, if we are certain that we want to follower to buy $\pi_0$, 
there is no obvious way to split $G$ into subgraphs. Our main technical contribution is a Split-Join lemma 
that shows that we can still split the tree into subtrees, each of which contain $\pi_0$, 
and solve each individually; the catch is that this only works if the leader can incentivize the follower 
to buy $\pi_0$ in every sub-instance. We call this requirement a {\em commitment to sell} or a 
{\em Yes-commitment}. The optimal solutions of the three 
sub-instances shown in Figure~\ref{subfigure-split-yes} are, in order from left to right:

\begin{enumerate}
    \item Setting $p(\pi_0) = 8$ and $p(\pi_1)=\infty$, 
    yielding a revenue of $8$.
    \item Setting $p(\pi_0)=0$ and $p(\pi_2) = 4$, earning 
    $4$.
    \item Setting $p(\pi_0)=10$ and $p(\pi_3)=3$, earning 13.
\end{enumerate}

We can now combine the pricing schemes of the three 
sub-instances  of Stackelberg-VC with commitment by adding 
up the pricing functions, obtaining $p(\pi_0) = 18$, $p(\pi_2)=4$, $p(\pi_3)=3$, and $p(\pi_1) = \infty$, achieving
a revenue of $25$.
Note that in the middle instance, we could earn 
$7$ by setting $p(\pi_0)=\infty$ and $p(\pi_2)=7$; however,
this would violate the commitment to sell $\pi_0$ in every 
instance, and indeed those three pricing schemes would not 
combine to a global one in a correct or meaningful way. 

In \cref{sec-def-tools} we formally define Stackelberg VC with commitment and state the 
Split-Join lemma. Its proof is based on LP duality and the integrality of the vertex cover LP 
on bipartite graphs.

In the above example we split $G$ into parts at a priceable vertex. If all priceable vertices are leaves, then we cannot split it further. Is it possible to split $G$ into sub-instances at a fixed-price vertex in a meaningful way? Yes, as it turns out, but this operation is much more expensive because we have to decide how the weight $w(u)$ should split into  weights in the sub-instances, and we have not found anything smarter than trying all possible ways. This is why our algorithm for general trees is only pseudo-polynomial. \\

\textbf{Pseudopolynomial and strongly polynomial algorithms.} An input instance to Stackelberg vertex cover contains {\em numbers} (integers or rationals). We count arithmetic operations like addition, subtraction, maximum, and minimum on $\mathbb{Q}$ as {\em one} step. Some of our algorithms are {\em strongly polynomial}: their running time is polynomial in $|V|$ and does not depend on the weights $w$; some other of our algorithms are only {\em pseudopolynomial}: the might depend on the largest weight seen in the input (for all our pseudopolynomial algorithms we have to assume that all input weights are integers). \\

\textbf{LCA trees and visibility.} Suppose $G = (P \dot{\cup} F, E)$ is a tree with some $r \in P$ selected as root. We call this an {\em LCA-tree} if the least common ancestor of any two priceable vertices is again priceable. 
\begin{definition}[Visibility]
\label{def:visibility}
Let $u, v \in V$. We say $u$ {\em can see} $v$ if, on the path from $u$ to $v$, all vertices except possibly $u$ and $v$ are fixed-price.  The {\em visibility} of $G$ is the smallest $k$ such that every $u \in F$ can see at most $k$ priceable vertices.     
\end{definition}

Note that LCA trees are exactly those trees of visibility at most $2$, which includes all paths. 

\begin{restatable}{theorem}{LCAResult}\label{thm:LCA-correct}
    The optimal value and the optimal pricing of the Stackelberg vertex cover game can be computed in strongly linear time on LCA-trees.
\end{restatable}

\begin{restatable}{theorem}{FPTResult}
    There is an algorithm solving Stackelberg vertex cover on trees of visibility $k$ running in $\mathcal{O}(2^k k n)$ time. This running time is strongly polynomial FPT with visibility as the parameter.
\end{restatable}

The previous technique allowed us to solve an instance by splitting it at every priceable vertex; this leaves us with small components that have at most $k$ priceable vertices each, with $k$ being the visibility number. As we'll see, one can solve such an instance in time $\mathcal{O}\left(k2^k n'\right)$ where $n'$ is the number of vertices in this component. Both of these results are shown in \cref{sec:LCA-trees}.\\

Somewhat orthogonal to these results, we present in \cref{sec-pseudopolynomial} a pseudopolynomial algorithm for instances with integer weights. This is again based on splitting and joining instances but now at fixed-price vertices $u$, too. As mentioned above, this is more complex because we have  to guess the weight of $u$ in the sub-instances. Let $W := 1+\max_{u \in F} w(u)$. 

\begin{restatable}{theorem}{PseudopolyResult}
\label{theorem-pseudopolynomial-tree}
    Stackelberg Vertex cover is solvable in $\mathcal{O}(W^3n)$ time for trees with integer weights in $\{0, \dots, W-1\}$.
\end{restatable}

This also works for bipartite graphs of bounded treewidth:

\begin{restatable}{theorem}{TreewidthResult}
    Stackelberg Vertex cover is solvable in $\mathcal{O}\left(nk2^k W^{2k+1}\right)$ for bipartite graphs if all weights are integers in $\{0, \dots, W-1\}$ and we are given a tree decomposition of width $k$.
\end{restatable}

This result is shown in \cref{sec:treewidth}. Note that our approach only holds for bipartite graphs as the solutions for general graphs do not compose similarly (more also in \cref{sec:treewidth}).\\

\textbf{Sell-all Stackelberg problems.} There is a variant of the Stackelberg problem where the leader has to design the price function such that the follower buys {\em all} vertices in $P$ (and possibly some in $F$). This problem is easier:

\begin{theorem}[{\cite[Proposition 1]{DBLP:journals/algorithmica/BriestHK12}}]
  Sell-all Stackelberg can be solved in polynomial time whenever the underlying one-player optimization problem (in our case: minimum weight vertex cover on bipartite graphs) can. 
\end{theorem}

Their algorithm builds a linear program with $|P|$ variables $2^{|P|}$ constraints. Still, if the one-player optimization problem can be solved efficiently, one can build a separation oracle and use the ellipsoid algorithm. The resulting algorithm is thus polynomial but not {\em strongly} polynomial. In Section~\ref{section-completely-committed} we give a strongly polynomial algorithm:

\begin{theorem}
\label{theorem-sell-all}
    Sell-all Stackelberg vertex cover on bipartite graphs can be solved in strongly polynomial time. The only arithmetic operations needed and addition, subtraction, max, and min. Our algorithm performs three calls to a maxflow algorithm on a network with $n+2$ vertices and $m+n$ edges. 
\end{theorem}

\begin{theorem}
    \label{theorem-sell-all-tree}
    Sell-all Stackelberg vertex cover on trees can be solved in $\mathcal{O}\left(|P|n\right)$ time.
\end{theorem}

\section{Definitions and Tools}\label{sec-def-tools}

We will often add up functions whose domain are neither disjoint nor identical. Given $f_i : A_i \rightarrow \nonneginf$ for $1 \leq i \leq k$, we define
\begin{align*}
f_1 + \cdots + f_k : \bigcup_{i=1}^k A_i & \rightarrow \R \cup \{\infty\} \\
              x & \mapsto \sum_{i: x \in A_i} f_i(x) \ , 
\end{align*}
where $c + \infty = \infty + c = \infty$ by convention. If the leader chooses a price function  $p: P \rightarrow \R^+_0 \cup \{\infty\}$ then $w + p$ is a total function $V \rightarrow \R \cup \{\infty\}$. The follower then chooses a minimum-weight vertex cover $X \subseteq V$ for $w + p$. Consider the set $P \setminus X$. These are the priceable vertices that the leader does not sell; thus, she could just as well set the weights to $\infty$, and  $X$ would still be a minimum-weight vertex cover. Therefore, there is no harm in requiring that leader strives to make follower buy all vertices in $P$ that have a finite price. 

\subsection{Feasible Pricing Schemes}

\begin{definition}
 A pricing scheme $p : P \rightarrow \R_0^+ \cup \{\infty\}$ 
 is called {\em feasible} for a graph $G$ and a weight function $w : V(G) \to \R_0^+\cup \{\infty\}$ 
 if there is a minimum weight vertex cover $X$
 of $G$ for the weight function $p + w: V \rightarrow \R_0^+ \cup \{ \infty \}$ with 
    \begin{align*}
    \forall u \in P: p(u) < \infty \rightarrow u \in X \ . 
    \end{align*}
    By $\feasible(G,w)$ we denote the set of pricing 
    schemes that are feasible for $G$ and $w$.
    The {\em revenue} of a feasible pricing scheme $p$ is 
    \begin{align*}
     \sum_{u \in P: p(u) < \infty} p(u) \ . 
    \end{align*}
\end{definition}

\begin{problem}[Stackelberg-VC]
    Given a graph $G = (V,E)$, a partition $V = P \cup F$, and a weight function $w: F \rightarrow \R_0^+ \cup \{\infty \}$,  find a feasible pricing scheme $p : P \rightarrow \R_0^+ \cup \{\infty\}$ of  maximum revenue.
\end{problem}

\subsection{Commitments}

In the course of our algorithms, we will encounter situations 
where the leader commits to making the follower buy 
certain vertices but not others. Formally, we define 
$\bar{V} := \{ \bar{v} \ | \ v \in V\}$. A {\em commitment}
is a set $C \subseteq V \cup \bar{V}$. It is understood that 
$C$ is non-contradictory, i.e., $v$ and $\bar{v}$ are never both in $C$. 
We call a $v \in C$ a {\em Yes-commitment to $v$} and $\bar{v} \in C$ 
a No-commitment to $v$. 
\begin{definition}
    A pricing scheme $p : P \rightarrow \R_0^+ \cup \{\infty\}$ is 
    {\em feasible for $G$, $w$ with commitment $C$} if 
    there exists a vertex cover $X \subseteq V$ of minimum cost 
    under $w + p$ such that 
    \begin{enumerate}
        \item $u \in X$ for all $u \in P$ with $p(u) < \infty$
        (this is just the previous notion of feasibility);
        \item $v \in X$ for all $v \in C$;
        \item $v \not \in X$ for all $\bar{v} \in C$.
    \end{enumerate}
    By $\feasible(G,w,C)$ we denote the set of all feasible 
    pricing schemes for $G$, $w$ with commitment $C$.
\end{definition}

A key insight of our work is that if $G$ has a cut vertex $u \in P$, i.e. a vertex $u$ such that removing $u$ and its edges increases the number of connected components,  then we can 
split $G$ at $u$ into parts and, for each part, solve a smaller Stackelberg-VC problem,
albeit with commitments as boundary conditions. Thus, we define:

\begin{problem}[Stackelberg-VC with Commitments]
    Given a graph $G = (P \cup F, E)$, a weight function $w: F\rightarrow \R_0^+$,
    and a set $C$ of commitments, find a pricing scheme $p \in \feasible(G,w,C)$
    of maximum revenue.
\end{problem}

\subsection{Solving committed instances}

In this section, we try to give an intuition for \cref{theorem-sell-all} and \cref{theorem-sell-all-tree}. The full proofs are in \cref{sec:computingPrices} in the appendix. If we set all prices to zero and compute the minimum weight vertex cover $X$ with $X \supseteq P$, then $X$ is still a minimum weight vertex cover for any pricing scheme $p$ that is feasible for $G, w$ with commitment $P$: $\feasible(G,w,P) = \feasible(G,w,X)$. 
We construct the a graph $G' = (V, E')$ with $E' = \{e \in E \mid |e \cap X| = 1\}$, i.e. by removing all doubly covered edges. In $G'$ all vertices with a yes commitment are now on the same side of the bipartition. Using duality, we can show that this modification preserves solutions.

\begin{restatable}{lemma}{bipartiteModification}
    $\feasible(G,w,X) = \feasible(G',w,X)$.
\end{restatable}

For a bipartite graph where all priceable vertices are on the same side, there is already a strongly polynomial-time algorithm in \cite{DBLP:journals/algorithmica/BriestHK12}. We can implement this using a constant number of maximum flow computations and using \cite{DBLP:conf/stoc/Orlin13} we get a runtime of $\mathcal{O}(nm)$.

For trees we can again be even faster by considering a linear program that computes optimal prices that sell all vertices. This is already described in \cite{DBLP:journals/algorithmica/BriestHK12} and is in principle of exponential size. Since now we want to sell all vertices on the same side, we can prove that we optimize over a polymatroid. For this class of problems there exist a simple greedy algorithm solving this in $\mathcal{O}(|P|n)$ time\cite[Chapter 44]{schrijver2003combinatorial}.

\subsection{Splitting Lemmas}

We describe how an instance of Stackelberg VC with commitments 
can be broken into smaller instances. We say {\em $G$ splits into $G_1, \dots, G_d$ at $u$}
if each $G_i = (V_i, E_i)$ is an induced subgraph of $G$ such that 
$V_1 \cup \dots \cup V_k = V$,  $E = E_1 \cup \dots \cup E_k$, and $V_i \cap V_j \{u\}$ for all $i \ne j$.\footnote{We technically allow $V_i = \{u\}$ in the split; but even if we forbid this, the split will not 
be unique when $2 \leq d < \deg(u)$. See the example: we could have split the central vertex 
into three parts differently.} 
For the rest of the section, $G$ always denotes a connected bipartite graph, 
$u$ a cut vertex of $G$, and $G_1, \dots, G_d$ are graphs into which $G$ splits at $u$. 
 
\begin{lemma}[Split-Join Lemma]
\label[lemma]{lemma-join-split}
    The following two statements hold:
    \begin{enumerate}
        \item \textbf{Join.} If for each $1 \leq i \leq d$ we have a weight function 
        $c_i : V(G_i) \to \R_0^+$ and a minimum weight vertex cover $X_i$ 
        for $G_i$ under $c_i$ such that either $u \in X_i$ for all $i$ 
        or $u \not \in X_i$ for all $i$, then 
        $X := \bigcup X_i$ is a minimum weight vertex cover for $G$
        under $c := c_1 + \cdots + c_d$. 
        \item \textbf{Split.} If $c: V \rightarrow \R_0^+$ is a weight function 
        and $X \subseteq V$ is a minimum weight vertex cover of $G$ under $c$, 
        then we can write  $c = c_1 + \cdots + c_d$ with $c_i : V(G_i) \rightarrow \R_0^+$ such that $X_i := X \cap V(G_i)$ is a minimum 
        weight vertex cover of $G_i$ under $c_i$.
        Furthermore, 
    if $c(v) \in \N$ for all $v \in V$, then we can choose 
    each $c_i$ such that $\im(c_i) \subseteq \N$, too.
    \end{enumerate} 
\end{lemma}

In words: minimum vertex covers can be joined and split between $G$ and $G_1, \dots, G_d$ as long 
as they agree on whether to include $u$.
\begin{figure}[H]  
\centering
\includegraphics[width=\textwidth]{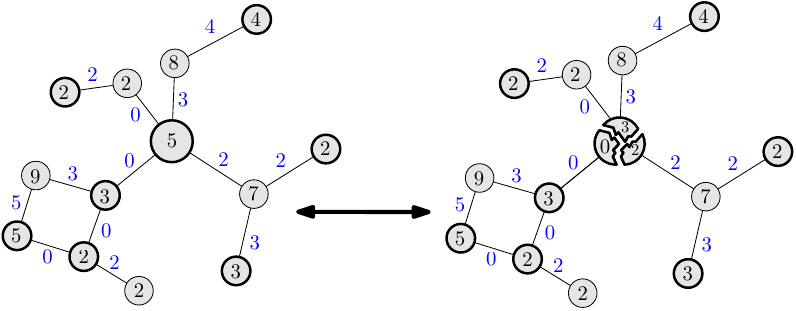}\\
\caption{Splitting and joining, using the dual 
solution to determine how the weight of $u$ should 
be split.}
\end{figure}
As a caveat, here is an example where 
joining the optimal vertex covers 
of the $G_i$ might give a suboptimal vertex cover for $G$
if the $X_i$ don't agree on $u$:
\begin{figure}[H]
\centering
\includegraphics[width=\textwidth]{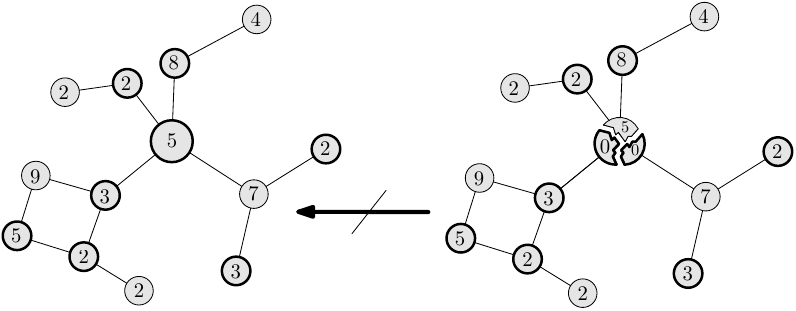}\\
\caption{The union of optimal vertex covers of the $G_i$ 
can be suboptimal for $G$ and $c_1 + \cdots + c_d$ 
if the $G_i$ don't agree on $u$.}
\end{figure}

\begin{proof}
    For Point 1, the join, let 
    let $y_i : E_i \rightarrow \R_0^+$ be an optimal
    dual solution to the weighted vertex cover linear program
    for $G, c_i$. We combine them to one big $y : E(G) \rightarrow \R_0^+$.    
    This $y$ is a feasible solution to the dual of the weighted vertex cover LP 
    with weight function $c$: if $v \ne u$ then $v$ belongs to exactly one graph $G_i$ 
    and 
    \begin{align*}
        \sum_{\substack{e \in E\\ v \in e}} y(e) = \sum_{\substack{e \in E_i\\ v \in e}} y_i(e) \leq c_i(v) ,
    \end{align*}
    For $v=u$ we have an additional sum on both sides:
    \begin{align*}
        \sum_{\substack{e \in E\\ u \in e}} y(e) = \sum_{i=1}^d \sum_{\substack{e \in E_i\\ u \in e}} y_i(e)
        \leq \sum_{i=1}^d c_i(u) = c(u) \ . 
    \end{align*}
    So $y$ is a feasible dual solution. 
    Since all $G_i$ are 
    bipartite, the minimum weight vertex covers $X_i$
    are optimal primal solutions. Since all $X_i$ agree on $u$, we have 
    \begin{align*}
        c(X) & = \sum_{v \in X} c(v) = \sum_{i=1}^d \sum_{v \in X_i} c_i(v) \\
        & = \sum_{i=1}^d \sum_{e \in E_i} y_i(e) \tag{by strong LP duality} \\
        & = \sum_{e \in E(G)} y(e) \ . 
    \end{align*}
    So the value of the primal solution $X$ and the dual solution $y$ agree and thus they 
    are both optimal.\\

    For the second part, the {\em split}, 
    take an optimal dual solution $y : E(G) \rightarrow \R_0^+$. Define $y_i := y|_{E_i}$.
    When choosing how to split $c$ into $c_1 + \cdots + c_d$, the only thing to decide 
    is how to set $c_i(u)$. The dual solution tells us how. Let $U_i$ be the set of neighbors of $u$ in $G_i$. We set 
    \begin{align*}
        c_i(u) := \sum_{v \in U_i} y( \{u, v\})  \ . 
    \end{align*}
    for $1 \leq i \leq d-1$ and 
    $c_d(u) = c(u) - c_1(u) - \cdots - c_{d-1}(u)$. Since 
    $y$ is a feasible dual solution we have 
    $\sum_{i=1}^d \sum_{v \in U_i} y_i(\{u, v\}) \leq c(u)$ and thus 
    $\sum_{v \in U_d} y_d(\{u, v\}) \leq c_d(u)$, too. For 
    all other $v \in V(G_i)$ we simply set $c_i(v) := c(v)$. 
    Thus, each 
    $y_i$ is a feasible dual solution for the minimum vertex cover 
    LP of $G_i$ with weights $c_i$.
    We have 
    \begin{align*}
        \sum_{e \in E} y(e) & = \sum_{i=1}^d \sum_{e \in E_i} y_i(e) \\
        & \leq 
        \sum_{i=1}^d \sum_{v \in X_i} c_i(v) \tag{by weak LP duality}\\
        & = 
        \sum_{v \in X} c(v)  \ , 
    \end{align*}
    where the last equality holds because all $X_i$ agree on $u$. The left-most 
    and right-most expression agree by strong duality, and the central inequality 
    holds for each $i$ individually; thus equality must hold for each $i$, 
    witnessing that $X_i$ is an optimal primal and $y_i$ is an optimal dual solution.\\
    
    Furthermore, if $c(v) \in \N$ for all $v \in V$, 
    then by the integrality of the matching LP for bipartite 
    graphs, we can choose the optimal dual solution $y$ to 
    consist of only natural numbers, too, and thus 
    $\im(c_i) \subseteq \N$.
\end{proof}

This lemma lets us split an instance of  Stackelberg-VC with commitments into multiple, smaller instances whenever $G$ has a cut vertex $u$:
\begin{lemma}[Splitting a Stackelberg instance at a committed vertex]
\label[lemma]{lemma-stackelberg-splitting}
    Let $G = (P \cup F, E)$ be connected and bipartite and 
    $u \in P \cup F$ a cut vertex, and suppose $G$ splits into $G_1, \dots, G_d$ at $u$.
    Set $P_i := P \cap V(G_i)$ and $F_i := F \cap V(G_i)$. 
    Let $C \subseteq V \cup \bar{V}$ be a set of commitments 
    containing one of $u$ and $\bar{u}$.
    Let $C_i \subseteq C$ be the commitments 
    belonging to vertices of graph $G_i$.
    \begin{enumerate}
        \item \textbf{Join.} If for each $1 \leq i \leq d$ we have a weight function $w_i : F_i \rightarrow \R_0^+$ 
        and a pricing 
        scheme $p_i : P_i \rightarrow \R_0^+$ that is feasible for $G_i, w_i$ with commitment 
        $C_i$, 
        then $p := p_1 + \cdots + p_d$ is feasible for $G$ 
        and $w := w_1 + \cdots + w_d$ with commitments $C$.
        \item \textbf{Split.} If 
        $w : F \rightarrow \R_0^+$ is a weight function and
        $p: P \rightarrow \R_0^+$ is feasible for $G, w$ with commitment 
        $C$, then we can split $p = p_1 + \cdots p_d$ 
        and $w = w_1 + \cdots + w_d$ such that  
        each $p_i$ is feasible for $G_i, w_i$ with commitment $C_i$.
        Furthermore, if $\im(w), \im(p) \subseteq \N$ then 
        we can choose the $w_i, p_i$ such that $\im(w_i), \im(p_i) \subseteq \N$, too.
    \end{enumerate}
\end{lemma}

\begin{proof}
    For the first part, the {\em join}, 
    by assumption we have vertex covers 
    $X_1, \dots, X_d$ such that $X_i$ is of minimum weight under $c_i := w_i + p_i$. 
    Furthermore, $X_i$ fulfills its commitments, meaning they all contain $u$ or all don't. 
    We can now apply the Join-part of 
    Lemma~\ref{lemma-join-split} and 
    see that $X := X_1 \cup \dots \cup X_d$ is optimal 
    under $c := c_1 + \cdots + c_d$. Furthermore, 
    it fulfills all commitments in 
    $C = C_1 \cup \dots \cup C_d$.
    Since $c = p + w$ it follows that the pricing 
    scheme $p$ is feasible for $G, w$ with commitment 
    $C$. \\

    For the split part, feasibility of $p$ means that there is a vertex cover $X$ of $G$ that is minimum weight under $c := p + w$ and fulfills all commitments. The split-part of Lemma~\ref{lemma-join-split} tells us that there is a way to write $c = c_1 + \cdots + c_d$ such that $X_i := X \cap V(G_i)$ is optimal under $c_i$. We set set $p_i := c_i|_{P_i}$ and $w_i := c_i|_{F_i}$. Each $X_i$ fulfills the commitment $C_i$ and thus $p_i$ is a feasible pricing scheme for $G_i, w_i$ with commitment $C_i$, as claimed. The ``furthermore'' part also follows directly  from the ``furthermore'' part of Lemma~\ref{lemma-join-split}.
\end{proof}

Lemma~\ref{lemma-stackelberg-splitting} works whether 
$u \in P$ or $u \in F$. However, this belies 
a stark algorithmic difference: if $u \in P$, then 
there is only one way to split $w$ into $w_1 + \cdots + w_d$ among 
the graphs $G_i$. Therefore, we get an algorithmically 
``easy'' splitting step:

\begin{corollary}
    If the cut vertex $u$ is priceable and $C$ is a commitment 
    containing either $u$ or $\bar{u}$, then 
    \begin{align}
        \opt(G, w, C) = \sum_{i=1}^d \opt(G_i,w_i,C_i) \ , 
        \label{eqn-split-at-priceable}
    \end{align}
    where $w_i := w|_{F_i}$ and $F_i := V(G_i)$. 
\end{corollary}

If, on the other hand, $u \in F$, then to even construct 
the sub-instances $(G_i, w_i)$, we need to guess the correct 
splitting of $w$ into $w_1, \dots, w_d$: 

\begin{corollary}\label[corollary]{col-split-fix}
    If the cut vertex $u$ is fixed-price and $C$ is a commitment
    containing either $u$ or $\bar{u}$, then
    \begin{align}
        \opt(G, w, C) = 
        \sup_{\substack{w_1, \dots, w_d \\ 
        w_i : F_i \rightarrow \R_0^+ \\
        w_1 + \cdots + w_d = w}} \sum_{i=1}^d \opt(G_i,w_i,C_i) \ , 
        \label{eqn-split-at-fixed-price}
    \end{align}
    where $F_i := V(G_i)$. 
\end{corollary}

There are infinitely many ways to split the weight of a fixed-price vertex in general for \cref{col-split-fix}. The ``furthermore'' part of Lemma~\ref{lemma-stackelberg-splitting} and the following lemma lets us replace the $\sup$ over real-valued $w_i$ with a $\max$ over $\N$-valued $w_i$, which results into a finite amount of splits.

\begin{lemma}\label[lemma]{lem:int-prices}
    Any Stackelberg vertex cover instance on a bipartite graph with integer weights for the fixed price vertices has an optimal solution where every priceable vertex receives an integer price (or infinity).
\end{lemma}

To prove the lemma, one could for example fix an optimal pricing scheme with a corresponding minimum weight vertex cover and an optimal dual solution and apply local-change arguments. Alternatively, one can brute-force over all $2^{|P|}$ total commitments on $P$ and run an algorithm for ensuing the sell-all Stackelberg instance. Those algorithms, which we describe in Section~\ref{section-completely-committed}, all produce integral solutions. 


\section{LCA trees and bounded visibility}\label{sec:LCA-trees}

We will give an $\mathcal{O}\left(2^k k n\right)$ algorithm for trees of visibility at most $k$, as defined in Definition~\ref{def:visibility}. We begin with $k=2$, i.e., LCA trees because they already illustrate the central idea of our algorithms. 

\begin{figure}
    \begin{subfigure}[t]{.45\textwidth}
        \centering
        \begin{tikzpicture}[
            level distance=14mm,
            level 1/.style={sibling distance=34mm},
            level 2/.style={sibling distance=18mm},
            level 3/.style={sibling distance=12mm},
            level 4/.style={sibling distance=8mm},
            every node/.style={
                circle,
                draw,
                minimum size=6mm,
                inner sep=1pt
            },
            edge from parent/.style={draw, thick}
        ]
        
        \node {$t$}
        child {
            node[fill=green!25] {$r$}
            child {
                node {$u_1$}
                child {
                    node[fill=green!25] {$u_2$}
                    child { node {$u_3$} }
                }
            }
            child {
                node {$v_1$}
                child {
                    node[fill=green!25] {$v_2$}
                }
            }
            child {
                node {$w_1$}
                child {
                    node {$w_2$}
                    child { node[fill=green!25] {$w_3$} }
                }
            }
        };
        
        \end{tikzpicture}
        \caption{An LCA-tree}
    \end{subfigure}
    \begin{subfigure}[t]{.5\textwidth}
        \centering
        \begin{tikzpicture}[
            level distance=14mm,
            level 1/.style={sibling distance=34mm},
            level 2/.style={sibling distance=18mm},
            level 3/.style={sibling distance=12mm},
            level 4/.style={sibling distance=8mm},
            every node/.style={
                circle,
                draw,
                minimum size=6mm,
                inner sep=1pt
            },
            edge from parent/.style={draw, thick}
        ]
        
        \node[fill=green!25]{$t$}
        child {
            node {$r$}
            child {
                node {$u_1$}
                child {
                    node[fill=green!25] {$u_2$}
                    child { node {$u_3$} }
                }
            }
            child {
                node {$v_1$}
                child {
                    node[fill=green!25] {$v_2$}
                }
            }
            child {
                node {$w_1$}
                child {
                    node {$w_2$}
                    child { node[fill=green!25] {$w_3$} }
                }
            }
        };
        
        \end{tikzpicture}
        \caption{Not an LCA-tree, because the least common ancestor of $u_2$ and $v_2$ is $r$ and $r$ is not priceable}
    \end{subfigure}

    \caption{An example and a counter-example of an LCA-tree. Priceable vertices are shown in green.}
\end{figure}


\subsection{Dynamic Programming for LCA trees.}

We will solve the Stackelberg-VC problem for LCA trees by splitting at every priceable vertex until we only get graphs with at most two priceable vertices. Afterwards, we will recombine these using dynamic programming, choosing optimal commitments as we go. To make this precise, we need a bit of terminology. First, we use $T$ instead of $G$ in this section because $G$ is a tree. Suppose we have chosen some priceable vertex $r$ to be the root of $T$. Let $u$ be a priceable vertex.

\begin{itemize}
    \item $T_u$ is the subtree of $T$ rooted at $u$, i.e., the 
    subgraph of $T$ induced by $u$ and all its descendants.
    \item $T_u^0$ is the subtree induced by $u$ and those fixed-priced 
    descendants of $u$ that see $u$ but no priceable vertex besides $u$.
    \item Let $u_1, \dots, u_d$ denote those children of $u$ that are not in $T_u^0$. 
    Note that all $u_i$ are fixed-price, since otherwise the leader has a monopoly
    on the edge $\{u, u_i\}$ and $\opt = \infty$. 
    Let $T_u^i$ denote the subtree induced by $u$, $u_i$, and 
    all descendants of $u_i$.
    \item Since $T$ is an LCA tree, $u_i$ sees exactly one 
    priceable vertex besides $u$, which we call $v_i$.
    Let $T_{u,v_i}$ denote the graph induced by the set of vertices $w$ that 
    can see both $u$ and $v_i$ (which includes $u$ and $v_i$).    
\end{itemize}

\begin{figure}

    \includegraphics[width=0.3\textwidth]{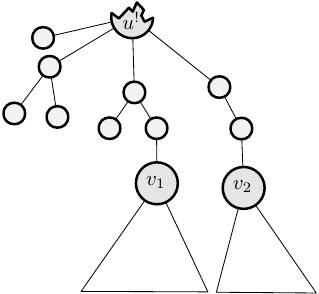}
    \includegraphics[width=0.3\textwidth]{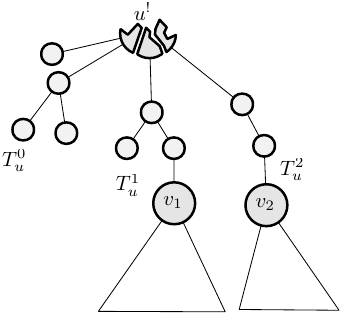}
    \includegraphics[width=0.3\textwidth]{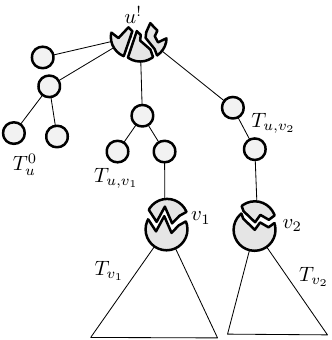}
    \caption{The tree $T_u$. The trees $T^0_u, T_u^1, \dots, T^d$. The trees 
    $T^0_u$ and $T_{u,v_i}$ and $T_{v_i}$ for $1 \leq i \leq d$.}
    \label{figure-Tu}
\end{figure}

See Figure~\ref{figure-Tu} for an illustration. 
Our dynamic programming algorithm computes a table $\rev$ where for each
$u \in P$ and commitment $u^? \in \{u, \bar{u}\}$, the value 
$\rev[u^?]$ stores the value $\opt(T_u, w, u^?)$. We will now describe 
how to compute $\rev[u^?]$, assuming that $\rev[v^?]$ has already been computed for 
all priceable descendants $v$ of $u$. 
By the Split-Join lemma, we have 
\begin{align*}
    \opt(T_u, w, u^?) = \sum_{i=0}^d \opt(T_u^i, w, u^?) \ . 
\end{align*}
Since the optimal pricing scheme for $(T^i_u, w, u^?)$ either sells $v_i$ or doesn't, we have 
\begin{align}
    \opt(T_u^i, w, u^?) = \max_{v_i^? \in \{v_i, \bar{v}_i\}} \opt(T_u^i, w, \{u^?, v_i^?\}) \ . 
    \label{eqn-commit-to-vi}
\end{align}
Once we have a commitment on $v_i^?$, we can split $T_u^i$ at $v_i$ into $T_{u,v_i}$ and $T_{v_i}$
and obtain
\begin{align*}
 \opt(T_u^i, w, \{u^?, v_i^?\}) = 
 \opt(T_{u,v_i}, w, \{u^?, v_i^?\}) + 
 \opt(T_{v_i}, w, \{u^?, v_i^?\}) \ . 
\end{align*}
The first instance, $(T_{u,v_i}, w, \{u^?, v^?_i\})$, has only two priceable 
vertices, each carrying a commitment. Its optimal value can be computed in linear 
time by the algorithm of Theorem~\ref{theorem-sell-all-tree}. The second instance, 
$(T_{v_i}, w, \{u^?, v_i^?\})$, does not contain $u$ anymore and thus can safely ignore 
the commitment $u^?$. Its optimal value can be looked up under $\rev[v_i^?]$. Therefore,
we can compute $\rev[u^?]$ by 
\begin{align*}
    \rev[u^?] = \opt(T_u^0, w, u^?) + \sum_{i=1}^d \max_{v_i^? \in \{v_i, \bar{v}_i\}} 
    \left( \opt(T_{u,v_i}, w, \{u^?, v_i^?\}) + \rev[v_i^?] \right)\  . 
\end{align*}

Finally, for the priceable root $r$ of $T$, we compute the optimum by 
\begin{align*}
    \opt(T,w) = \max_{r^? \in \{r, \bar{r}\}} \opt(T, w, r^?) = 
    \max_{r^? \in \{r, \bar{r}\}} \rev[r^?] \ . 
\end{align*}

\LCAResult*

All paths are LCA tress and thus our algorithm offers an alternative to the algorithm by Eickhoff, Kauther and Peis~\cite{DBLP:conf/sagt/EickhoffKP23}. \cref{fig-path-example} in the appendix contains an illustration of how our algorithm computes the optimal revenue on a path.

\subsection{Trees with bounded visibility}

Suppose $T$ has constant {\em visibility} $k$: every fixed-price vertex can see at most $k$ of priceable vertices. A very similar algorithm also works. What changes is (\ref{eqn-commit-to-vi}): child $u_i$ of $u$ sees up to $k-1$ priceable vertices besides $u$; call this set $V_i$. Then (\ref{eqn-commit-to-vi}) would become 
\begin{align}
    \opt(T_u^i, w, u^?) = \max_{C \textnormal{ on } V_i} 
    \opt(T_u^i, w, \{u^?\} \cup C) \ ,  
\end{align}
where the maximization is over all $2^{|V_i|-1} \leq 2^{k-1}$ commitments $C$ to the priceable 
vertices in $V_i$. The problem $\opt(T_u^i, w, \{u^?\} \cup C)$ then splits into $\opt(T_v, w, v^?)$ for each $v^? \in C$ plus one instance containing $u$, $V_i$, and the fixed-price vertices between them. The former can be looked up in $\rev$, the latter is fully committed and we can compute the optimal prices in $\mathcal{O}(kn)$ by Theorem~\ref{theorem-sell-all-tree}.

\FPTResult*

\section{Stackelberg-VC on trees with integer weights}\label{sec-pseudopolynomial}

While the previous algorithms work for any weights, we will now restrict ourselves to the integer case.
We believe that this is a natural restriction as most practical problems can be stated in integers or at least be converted into integers without a huge blow-up. 
We now want to split also fixed-price vertices using \cref{col-split-fix}. By \cref{lem:int-prices}, we know that it suffices to only split the weight into integer weights, because there is an optimal solution with integer prices which implies that there is an integral dual solution. The resulting algorithm  will only be pseudopolynomial.

\PseudopolyResult*

We again use a dynamic programming approach similar to \cref{sec:LCA-trees}, now splitting it at every vertex. Our base cases are thus single edges $\{u_1,u_2\}$ in which both endpoints carry a commitment and, if in $F$, have a weight. The revenue of this instance  can be expressed as a function $f(w_1, w_2, c_1, c_2)$ where $w_i \in \nonneginf \cup \{*\}$ is the price of $u_i$ if $u_i \in F$ or is $*$ if $u_i \in P$, and $c_i \in \{u_i, \bar{u}_i\}$ is the commitment. It holds that 
\begin{align*}
    f(w_1,w_2, \bar{u}_1, \bar{u}_2) = -\infty \quad \forall\  w_1, w_2 \in \nonneginf \cup \{*\} \ . 
\end{align*}
since no vertex cover can satisfy both no-commitments. For $w_1, w_2 \in \nonneginf$ (that is, both vertices are fixed price), we have 
\begin{align*}
    f(w_1, w_2, u, \bar{v}) & = \begin{cases}
        0 & \textnormal{ if } w_1 \leq w_2\\
        -\infty & \textnormal{ else.}
    \end{cases}\\
    f(w_1, w_2, u, v) & = 
    \begin{cases}
        0 & \textnormal{ if } w_1 = w_2 = 0,  \\
        -\infty & \textnormal{ else \ , } 
    \end{cases}
\end{align*}
since no minimum vertex cover of $\{u_1, u_2\}$ contains both vertices, unless both cost 0. If exactly one vertex is priceable, say $u_2$, we have 
\begin{align*}
    f(w_1,*, u_1, \bar{u}_2) & = 0 \\
    f(w_1,*, \bar{u}_1, u_2) & = w_1 \\
    f(w_1,*, u_1, u_2) & = \begin{cases}
        0 & \textnormal{ if } w_1 = 0  \ ,  \\
        -\infty & \textnormal{ else.} 
    \end{cases}
\end{align*}
The case $f(*,*,\cdot, \cdot)$ never happens since $P$ forms an independent set in $G$ (otherwise leader has a monopoly and can extract arbitrarily high revenue). The value of $f$ can be computed in constant time.

Pick an arbitrary vertex of the tree $G$ as root. For a vertex $u$ let $T_u$ be the subtree rooted at $u$, and let  $\Children(u)$ be the set of children of $u$. For a vertex $u \in V(G)$, a number $x \in \N$, an index $0 \leq i \leq |\Children(u)|$,  and a commitment $u^? \in \{u, \bar{u}\}$, we define a subproblem $(T_{u,i}, w^x_{u,i}, u^?)$ as follows:

\begin{enumerate}
\item Obtain $T_{u,i}$ by removing all incident edges to $u$ except those
to its first $i$ children; then let $T_{u,i}$ be the connected component containing $u$.
\item Obtain $w^x_{u,i}$ by having each $v \in V(T_{u,i})$ inherit its 
weight from the original instance $(G,w)$---except possibly $u$ 
which, if fixed-price, gets assigned weight $x$.  
\end{enumerate}

\noindent We will compute a table $\rev$ in which entry $\rev[u,x,i,u^?]$ hold the value $\opt(T_{u,i}, w^x_{u,i}, \{u^?\})$. The entry $\rev[u,x,i,u^?]$ can be computed from ``lower'' entries---those of vertices $v$ below $u$ or for $i' < i$.  Let $v_1, \dots, v_i$ be the $i$ children of $u$ in $T_{u,i}$  and $T_1, \dots, T_i$ the subtrees rooted there. For $0 \leq y \leq x$, $0 \leq z \leq w(v_i)$, and a commitment $v_i^? \in \{v_i, \bar{v_i}\}$, we use the Split-Join-lemma to split $(T_{u,i}, w_{u,i}, \{u^?\})$ at $u$ and then at $v_i$ into three instances:

\begin{enumerate}
    \item $u$ together with its first $i-1$ children and their subtrees with $u$ having 
    weight $y$: $(T_{u,i-1}, w^y_{u,i-1}, \{u^?\})$. 
    The optimal value of this instance is $\rev[u,y,i-1, u^?]$ and 
    can be looked up from the table in $\mathcal{O}(1)$ time.
    \item $u$ with its $i$th child $v_i$: formally the instance
    on graph $(\{u,v_i\}, [u \mapsto x-y, v \mapsto w(v_i) - z], \{\{u^?,v_i^?\}\})$.
    The optimal value of this instance is the function 
    $f(x-y, w(v_i)-z, u^?, v_i^?)$, which we compute in $\mathcal{O}(1)$ (replace $x-y$ and $w(v_i)-z$ by a $*$ if the involved vertices are priceable). 
    \item $v_i$ with its subtree and $v_i$ having weight $z$ and 
    commitment $\{v_i^?\}$. The value of this instance is 
    $\rev[v_i,z,|\Children(v_i)|, v_i^?]$ and can be looked up from the table.
\end{enumerate}

By the Split-Join lemma, this gives us the optimal value if we guess $y$ and $z$ correctly; otherwise, it gives us a value that is less than the optimum. Thus, 
we can compute $\rev[u,x,i,l]$ in $\mathcal{O}(w(u) \cdot w(v_i))$ time via 
\begin{align*}
    \rev[u,x,i,u^?] & = 
    \max_{v_i^? \in \{v, \bar{v}\}} \max_{0 \leq y \leq x} 
    \max_{0 \leq z \leq w(v_i)} 
    \left(
    \begin{array}{ll}
    & \rev[u,y,i-1, u^?]\\
    + & f(x-y, w(v_i)-z, u^?, v_i^?)\\
    + &  \rev[v_i,z,|\Children(v_i)|, v_i^?]
    \end{array}
     \right)
\end{align*}

The following figure illustrates how the instance $(T_{u,i}, x, i, u^?)$ is split into three 
smaller instances:
\begin{figure}[H]
    \centering
    \includegraphics[width=\textwidth]{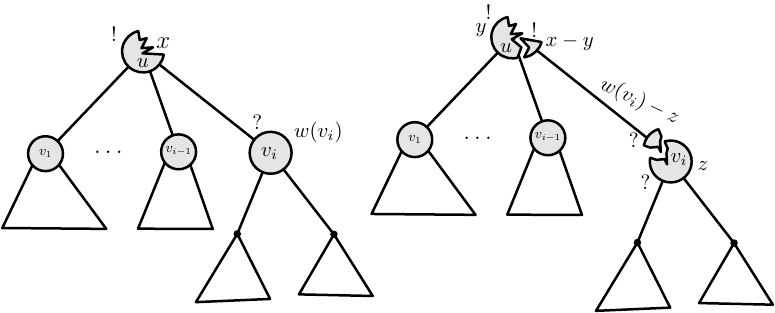}
    \caption{Computing $\Rev[u,x, i, u^?]$ by splitting into three instances.}
\end{figure}

We reach the base case for $i=0$; meaning $u$ has already lost all children, or has been a leaf to begin with. Then the instance $T_{u,i}$ is a single vertex, and one can easily see that 
\begin{align*}
    \rev[u, x, 0, u^?] = \begin{cases}
        - \infty\, &\text{, if $x> 0$ and $u^? = u$}\\
        0\, &\text{, if $x=0$ or $u^? = \bar{u}$.}\\
    \end{cases}
\end{align*}

Our table has $\sum_{v\in V} (|\children(v)|+1) \cdot W \cdot 2 = O(nW)$ many table entries. For each entry  we compute the optimal value in $\mathcal{O}(W^2)$ steps, giving a total runtime of $\mathcal{O}(nW^3)$, which concludes the proof.

\section{Graphs of bounded treewidth}\label{sec:treewidth}

Many problems that are efficiently solvable on trees can also be efficiently solved for graphs that are similar to trees. One way to quantify the similarity of a graph to a tree is the concept of {\em treewidth}.

\begin{definition}
    A {\em tree decomposition} of a graph $G= (V,E)$ is a tree $\decomptree = (\bags, \decompedges)$ with
    \begin{enumerate}
        \item $\bigcup_{\bag \in \bags} \bag = V$,
        \item for all $\{u, v\} \in E$ there exists some $\bag \in \bags$ with $u, v \in \bag$,
        \item for every vertex $v \in V$, the set of $\{\bag \in \bags \ | \ v \in \bag\}$ is a connected subgraph of $\decomptree$.
    \end{enumerate}
    The sets in $\bags$ are called {\em bags} and the {\em width} of $\decomptree$ is defined as $\max_{\bag \in \bags} |\bag| - 1$. The {\em treewidth} of $G$ is now the minimum width among all tree decompositions of $G$.
\end{definition}

We note that there are fixed-parameter algorithms computing a tree-decomposition~\cite{DBLP:journals/siamcomp/Bodlaender96, DBLP:conf/focs/Korhonen21} and assume for the remainder of this section that a tree decomposition is already available. One can easily modify $\decomptree$ to ensure that for every edge $\{\bag,\bag'\} \in \decompedges$ the set $\bag \cap \bag'$ is a separator of $G$. We generalize our Split-Join lemma from cut vertices to separators:

\begin{restatable}{lemma}{SplitJoinSeparator}
    \label{lem:split-join-sep}
    Let $G_i = (V_i, E_i)$ for $1 \leq i \leq d$ be graphs such that $V_i \cap V_j = S$ for all $1 \leq i < j \leq d$ and let $G = \left(\bigcup_{i} V_i, \bigcup_{i} E_i\right)$. Then the two statements hold: 
    \begin{enumerate}
        \item Suppose for each $1 \leq i \leq d$ we have a cost function $c_i : V_i \rightarrow \nonneginf$ and a minimum weight vertex cover $X_i$ for $G_i$ under $c_i$ such that the $X_i \cap S$ are identical for all $i$. Then $\bigcup_{i=1}^d X_i$ is a minimum weight vertex cover for $G$ under $c_1 + \dots + c_d$.
        \item If $c: V \rightarrow \nonneginf$ is a weight function and $X \subseteq V$ is a minimum weight vertex cover of $G$ under $c$, then we can write  $c = c_1 + \cdots + c_d$ with $c_i : V_i \rightarrow \nonneginf$ such that each $X_i := X \cap V_i$ is a minimum weight vertex cover of $G_i$ under $c_i$. Furthermore, if $c(v) \in \N$ for all $v \in V$, then we can choose each $c_i$ such that $\im(c_i) \subseteq \N$, too.
    \end{enumerate}
\end{restatable}

Now, we can adapt the algorithm of \cref{sec-pseudopolynomial} to work on the tree of the decomposition. Given an edge between two bags $B_i$ and $B_j$ such that $B_j$ is a child of $B_i$, we have to brute-force over all possible commitments to the vertices in $B_j \setminus B_i$ and all ways to split the weights of fixed-price vertices in $B_i \cap B_j$. 

\TreewidthResult* 

More details can be found in \cref{sec:treewidthApp}. Although we can relax the requirement of $G$ being a tree to $G$ having bounded treewidth, we still crucially need $G$ to be bipartite! In fact, the Split-Join lemma and the integrality lemma fail for non-bipartite graphs, as \cref{fig:counterexampleSplitJoin} illustrates: the optimal pricing scheme sells  priceable vertices for $1.5$ each, breaking integrality. To see that the Split-Join lemma fails, put a yes-commitment on $\{\pi_1, \pi_2, \pi_3\}$ and split the graph along this separator into the outer graph $G_1$, (a $C_6$) and the inner graph $G_2$ (a triangle plus $\pi_1, \pi_2,\pi_3$ as leaves). In $G_1$ leader earns $3$ by setting all prices to $1$. In $G_2$, the only feasible pricing scheme is $(0,0,0)$ since each $\pi_i$ carries a yes-commitment. Thus, 
\begin{align*}
    \opt(G_1, w_1, P) + \opt(G_2, w_2, P) = 3 + 0 < 4.5 = \opt(G, w_1 + w_2, P) \ ,
\end{align*}

meaning splitting a non-bipartite graph can decrease the optimal revenue.

\begin{figure}
    \centering
    \includegraphics[width=.4\linewidth, page=2]{img/CounterExampleSplitJoin.pdf}
    \caption{A non bipartite graph where the split join lemma for the separator $\{\pi_1, \pi_2, \pi_3\}$ and integrality fail.}
    \label{fig:counterexampleSplitJoin}
\end{figure}
\section{Conclusion}

In this work we have extended the algorithmic toolbox for Stackelberg-Vertex-Cover, which previously allowed to solve the problem exactly if all priceable vertices were on the same side~\cite{DBLP:journals/algorithmica/BriestHK12} (and improved by Ba{\"i}ou and Barahona~\cite{DBLP:journals/algorithmica/BaiouB16}) and if the graph is a path~\cite{DBLP:conf/sagt/EickhoffKP23}. Now we have an exact algorithm for LCA trees, a specialized parameterized algorithm compared to the general parameterized Stackelberg algorithm~\cite{DBLP:conf/icalp/BohnleinKS17} and  pseudopolynomial algorithms for bipartite instances with integer weights and bounded treewidth. While the proofs of correctness rely on linear programming duality, the algorithms are fully combinatorial.

Together with a currently unpublished \NP{}-completeness result by Eickhoff, Kauther and Peis~\cite{NPCTree}, this shows weak-\NP{}-completeness for the Stackelberg Vertex-Cover problem on trees and settles the general complexity question. Except for paths and their generalization LCA trees all algorithms were pseudopolynomial. Are there more classes that are solvable in polynomial time? Potentially interesting classes could be lobsters or caterpillars.

A curious fact about \cref{theorem-pseudopolynomial-tree} is that it uses two properties of trees: (1) they have small treewidth; (2) they are bipartite. Both are crucial for our algorithm to work. Indeed, it is an open question whether there is a pseudopolymial algorithm for non-bipartite graphs of small treewidth. As the example from \cref{fig:counterexampleSplitJoin} shows, a potential combination step has to be more involved, since the optimal revenue before the split is greater than after, in stark contrast to the bipartite case.

Outside of the Stackelberg-Vertex Cover problem it might be possible to reuse the technique of using the dual variables to split and combine smaller subproblems to other Stackelberg problems based on decision problems with a natural dual solution to obtain polynomial or pseudopolynomial algorithms.

\section{Use of AI}

In Section~\ref{np-hard-with-commitment} we show that Stackelberg VC with one commitment is NP-complete for trees. Using LLMs (ChatGPT), we have tried to modify the reduction such as to not use any commitment. Alas, without success. We asked Claude to prove or disprove whether $\feasible(Ax \leq b_1 + b_2)$, the set of feasible solutions to that LP, equals the Minkowsky sum $\feasible(Ax \leq b_1) + \feasible(Ax \leq b_2)$ (the relation $\supseteq$ always holds). Claude gave us a very simple counter-example with two variables. Understanding this ultimately led to the example in Figure~\ref{fig:counterexampleSplitJoin}.

\bibliography{main}

\appendix

\section{Remaining proofs for graphs of bounded treewidth}\label{sec:treewidthApp}

We will now extend the previous algorithms to the class of graphs of bounded tree-width. These graphs behave somewhat treelike and often allow to solve problems efficiently even if they are hard for general graphs. We repeat the definition of a tree-decomposition and the treewidth for the reader.

\begin{definition}
    A \emph{tree decomposition} of a graph $G= (V,E)$ is a tree $\decomptree = (\bags, \decompedges)$ with
    \begin{enumerate}
        \item $\bigcup_{\bag \in \bags} \bag = V$,
        \item for all $\{u, v\} \in E$ exists some $\bag \in \bags$ with $u, v \in \bag$, 
        \item for every vertex $v \in V$, the set $\{\bag \in \bags \ | \ x \in \bag\}$ of bags that contain $v$ form a connected subgraph of $\decomptree$.
    \end{enumerate}
    The sets in $\bags$ are called \emph{bags} and the \emph{width} of $\decomptree$ is defined as $\max_{\bag \in \bags} |\bag| - 1$. The \emph{tree-width} of $G$ is now the minimum width among all tree decompositions of $G$.
\end{definition}

A tree decomposition of width $k$ can be computed exactly in $2^{\mathcal{O}(k^3)}n$ if it exists\cite{DBLP:journals/siamcomp/Bodlaender96}. Alternative one can approximate this solution by computing a tree decomposition whose width is at most $2k+1$ where $k$ is the tree width in $2^{\mathcal{O}(k)}n$ time\cite{DBLP:conf/focs/Korhonen21}.

In contrast to cut-vertices we will now talk about separators. A \emph{separator} in a graph $G=(V,E)$ is a set of vertices $S \subseteq V$ such that the induced subgraph $G[V \setminus S]$ has more connected components than $G$. Specifically, any cut-vertex is a separator. 

A tree-decomposition is called \emph{reduced}~\cite{DBLP:journals/dam/LiMNS18} if there are no two bags $\bag_1, \bag_2$ with $\bag_1 \subseteq \bag_2$. The number of bags is at most $|V(G)| - 1$ if $G$ is a graph with at least two vertices\cite{DBLP:journals/dam/LiMNS18}. Given a tree decomposition we can transform it into a reduced tree decomposition in polynomial time by checking for every edge the intersection of the edges bags and merging a possible subset into its superset.
\begin{lemma}[folklore]
Let $\decomptree = (\bags, \decompedges)$ be a reduced tree decomposition of a connected graph $G=(V,E)$. Let $\bag_1, \bag_2 \in \bags$ be two bags in $\decomptree$ that are neighboring, i.e. $\{\bag_1, \bag_2\} \in \decompedges$. Then $\bag_1 \cap \bag_2$ is a separator of $G$.
\end{lemma}
\begin{proof}
    Imagine we remove the edge $\{\bag_1, \bag_2\}$. This splits the tree into two trees with bag sets $\bags_1$ and $\bags_2$, respectively. Let $V_i := \bigcup_{\bag \in \bags_i} B$. We claim that $V_1 \cap V_2 = \bag_1 \cap B_2$. The inclusion $\supseteq$ is obvious. For the other direction, let $x \in V_1 \cap V_2$, so there are bags $\bag'_1, \bag'_2$ with $x \in \bag'_i \in \bags'_i$ for $i=1,2$. The set of in $\bags$ containing $x$ forms a connected subgraph of $\decomptree$ by the third property of tree decompositions and thus contains a path from $\bag'_1$ to $\bag'_2$. Since $\decomptree$ is a tree, this path is unique and must contain the edge $\{\bag_1, \bag_2\}$. Thus, $x \in \bag_1 \cap \bag_2$.

    Now we claim that (1) there are $x_1 \in V_1 \setminus V_1$ and $x_2 \in V_2 \setminus V_2$ and (2) every path in $G$ from some $x_1 \in V_1$ to some $x_2 \in V_2$ passes through some vertex in $\bag_1 \cap \bag_2$. By reducedness we can find some $x_1 \in \bag_1 \setminus \bag_2$ and $x_2 \in \bag_2 \setminus \bag_1$. So $x_1 \in V_1$ and $x_2 \in V_2$ but neither is in $V_1 \cap V_2 = \bag_1 \cap \bag_2$. This shows (1). 

    To show (2) we consider a path in $G$ from $x_1$ to $x_2$. This path exists because $G$ is connected. Let $e_1, \dots, e_l$ be the edges on the path; for each edge pick some bag $\bag(e_i) \in \bags$ with $e_i \subseteq \bag$. Consider the last edge covered by a bag in $\bags_1$, i.e., the highest index $i$ for which $\bag(e_i) \in \bags_1$. If $i=l$ then $x_2 \in e_l \subseteq \bag(e_l) \subseteq V_1$ and thus $x_2 \in V_1 \cap V_2$. Otherwise consider $e_i \cap e_{i+1} = \{y\}$. We have $y \in e_i \subseteq B(e_i) \subseteq V_1$ and $y \in e_{i+1} \subseteq B(e_{i+1}) \subseteq V_2$ so $y \in V_1 \cap V_2$. In either case, the path passes through $V_1 \cap V_2 = \bag_1 \cap \bag_2$.    
\end{proof}

\begin{definition}
Let $G = (V,E)$ be a graph and $S$ a separator. We say $G$ {\em splits into $G_1, \dots, G_d$ at $S$} if 
\begin{enumerate} 
\item each $G_i = (V_i, E_i)$ is an induced subgraph of $G$, 
\item $V_1 \cup \dots \cup V_d = V$ and $V_i \cap V_j  = S$ for all $1 \leq i < j \leq d$.
\end{enumerate}
\end{definition}

\SplitJoinSeparator*

The proof is slightly more involved than that of Lemma~\ref{lemma-join-split} because there might now be edges inside the separator $S$ (which isn't the case if $S$ is a cut vertex) and thus we have to think a bit about how to distribute edge weights inside $S$ among the subgraphs; this is a nuisance but not an obstacle. Figure~\ref{fig:splitjoinseparator} illustrates the proof of the lemma.

\begin{figure}
    {\centering
    \includegraphics[width=0.8\linewidth]{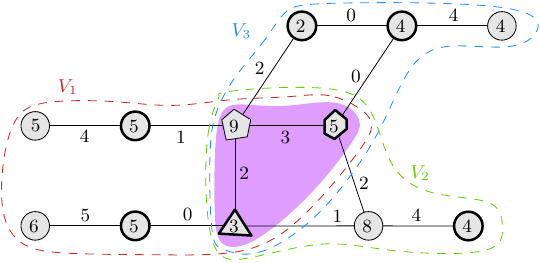}\\
    \includegraphics[width=1\linewidth]{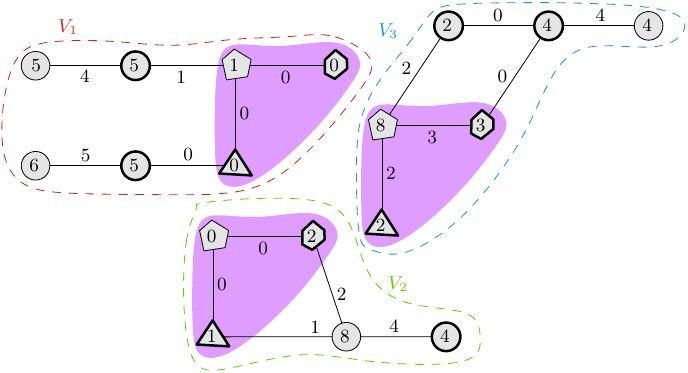}
    }
    \caption{A graph with a separator before and after the split. The split is not unique: $G - S$ actually has four connected components, but in our example two of them end up in $V_1$. This is totally allowed in the  Split-Join lemma. Graph $V_3$ inherits all edge weights inside $S$. The edge weights in the subgraphs tell us exactly how to distribute vertex weights in $G_1, \dots, G_{d-1}$. The graph $G_d$ then takes whatever weight is left.}
    \label{fig:splitjoinseparator}
\end{figure}

\begin{proof}
    \begin{enumerate}
        \item We consider the vertex cover $X = \bigcup_{i=1}^d X_i$ and show its optimality for $c = c_1 + \dots + c_d$. First we have 
        \begin{align*}
        c(X) & = c(X \setminus S) + c(X \cap S) = 
        c\left(\bigcup_{i} (X_i\setminus S) \right) + c(X \cap S) \\
        & = \sum_{i} c(X_i \setminus S) + c(X \cap S) \tag{by disjointness} \\
        & = \sum_{i} c_i(X_i \setminus S) + c(X \cap S) \tag{$c_i$ and $c$ agree on $V_i \setminus S$}\\
        & = \sum_{i} c_i(X_i \setminus S) + \sum_{i}c_i (X \cap S) \tag{definition of $c$}\\
        & = \sum_{i} c_i(X_i \setminus S) + \sum_{i}c_i (X_i \cap S) \tag{since $X \cap S = X_i \cap S$}\\
        & = \sum_i c_i(X_i) \ . 
        \end{align*}
        
        Additionally, since all $X_i$ is an optimal solutions, we have a dual solutions $y_i : E_i \rightarrow \nonneginf$ of the same weight for each $i$. We obtain the new dual solution via $y = y_1 + \dots + y_d$. This is a feasible solution, since for any vertex $v$ :
        \begin{equation*}
            \sum_{e \in E, v \in e} y(e) =
            \sum_{i=1}^d \sum_{e \in E_i, v \in e} y_i(e)
                \le \sum_{i=1}^d c_i(v)
                = c(v)
        \end{equation*}

        Finally, $\sum_{e \in E} y_e = \sum_{e \in E} \sum_{i=1}^d y_i(e) = \sum_{i=1}^d \sum_{e \in E_i} y_i(e) = \sum_{i=1}^d c(X_i) = c(X)$. So the primal solution $X$ and the dual solution $y$ have the same value; by weak duality both are optimal, and thus $X$ is a minimum weight vertex cover for $G$ under $c$.
        
        \item We are given an optimal vertex cover $X$ for $G$. By strong LP duality there is a dual solution $y : V \rightarrow \nonneginf$ with $\sum_e y(e) = c(X)$. Let $E_S := E \cap \binom{S}{2}$ be the edges within $S$. Since $S$ is a separator the $V_i \setminus S$ are disjoint and so are the $E_i \setminus E_S$. We define dual solutions $y_i$ via         
        \begin{equation*}
            y_i(e) = \begin{cases}
                y(e) & \text{if $e \in E_i \setminus S$ or $e \in E_S$ and $i=d$,}\\
                0 & \text{else.}
            \end{cases}
        \end{equation*}
        In words: for $e \in E_i \setminus S$ we assign its weight $y(e)$ to $y_i$; for $e \in E_S$ we assign it to $y_d$. For ease of notation we extend $y_i$ to all of $E$ by setting $y_i(e)=0$ for all $e \not \in E_i$. With this convention we have $y(e) = \sum_i y_i(e)$ for every edge $e$. We define the cost functions $c_i$ as follows: for $v \in V_i \setminus S$ we simply set $c_i(v) := c(v)$. For $v \in S$ we set 
        \begin{align}
            c_i(v) & = \sum_{e \in E_i: v \in e} y_i(e) \qquad \textnormal{for $1 \leq i \leq d-1$} 
            \label{eqn-def-c_i}\\
            c_d(v) & = c(v) - \sum_{i=1}^{d-1} c_i(v) \ . 
            \label{eqn-def-c_d}
        \end{align}
        Clearly $c = c_1 + \cdots + c_d$. Each $X_i = X \cap S$ is a vertex cover for $G_i$. Each $y_i$ is a feasible dual solution: for $v \in V_i \setminus s$ this follows because 
        \begin{align*}
            c_i(v) & = c(v) \geq \sum_{e \in E: v \in e} y(e) \tag{ because $y$ is feasible} \\
                & = \sum_{e \in E_i \setminus E_S: v \in e} y(e) \tag{ because they all are in $E_i \setminus E_S$} \\
               & = \sum_{e \in E_i \setminus E_S: v \in e} y_i(e) \\
               & = \sum_{e \in E_i: v \in e} y_i(e)
        \end{align*}
        For $v \in S$ and $i \leq d-1$ this holds with tightness by (\ref{eqn-def-c_i}). For $v \in S$ and $i=d$ this requires a short computation.  Then
        \begin{align*}
            c_d(v) & = c(v) - \sum_{i=1}^{d-1} c_i(v)  = c(v) - \sum_{i=1}^{d-1} \sum_{e \in E_i: v \in e} y_i(e)\\
            & = c(v) -  \sum_{e \in E: v \in e} \sum_{i=1}^{d-1} y_i(e)\\
            & = c(v) -  \sum_{e \in E: v \in e} \left(y(e) - y_d(e)\right)\\
            & = c(v) - \sum_{e \in E} y(e) + \sum_{e \in E_d} y_d(e) \\ 
            & \geq \sum_{e \in E_d} y_d(e) \tag{by feasibility of $y$}  \ . \\ 
        \end{align*}
        So each $y_i$ is feasible and thus all $c_i$ are non-negative. We have already seen that each $X_i$ is a vertex cover of $G_i$. The fact $c(X) = \sum_i c_i(X)$ follows by design from (\ref{eqn-def-c_d}). We get 
        \begin{align*}
            c(X) & = \sum_{i=1}^d c_i(X_i) \tag{follows from (\ref{eqn-def-c_d}}\\
            & \geq \sum_{i=1}^d \sum_{e \in E_i} y_i(e) \tag{form weak duality}\\
            & = \sum_{e \in E} \sum_{i=1}^d y_i(e) \\
            & = \sum_{e \in E} y(e) \ . 
        \end{align*}
        By strong duality, $c(X) = \sum_{e \in E} y(e)$. Since the inequality above holds for each $i$ individually, equality has to hold throughout: $c_i(X_i) = \sum_{e \in E_i} y_i(e)$. Weak duality now implies that $X_i$ and $y_i$ are optimal solutions. In particular, $X_i$ is a minimum weight vertex cover for $G_i$ and $c_i$.\\
        
        Furthermore, we see from the construction that all solutions preserve the integrality of the weights and dual solutions.
    \end{enumerate}
\end{proof}

\subsection{DP algorithm for bipartite graphs of bounded treewidth}

Our algorithm is similar to the one in \cref{sec-pseudopolynomial}. Let $\decomptree = (\bags, \decompedges)$ be a reduced tree-decomposition of a bipartite graph $G = (V,E)$. The weight function $w$ for $G$ maps every fixed-price vertex to an integer weight. We choose an arbitrary bag $\bag_{\rm root}$ to be the root of $\decomptree$. For a $\bag \in \bags$ let $\children(\bag)$ denote the children of $\bag$ in this rooted tree. We in \cref{sec-pseudopolynomial} we need some notation to talk about subtrees of $\decomptree$, corresponding subgraphs of $G$, and sub-instances of the Stackelberg VC with commitments. 
\begin{enumerate}
    \item For bag $\bag \in \bags$ and $0 \leq i \leq |\children(\bag)|$ let $\decomptree_{\bag,i}$ be the tree that is obtained by removing all incident edges to $\bag$ in $\decomptree$ except those to its first $i$ children and taking the connected component containing $\bag$. We view this as a tree rooted at $\bag$.
    \item Let $V_{\bag,i} := \bigcup_{\bag' \in \mathcal{V}(\decomptree_{\bag,i})} \bag'$ be the set of all vertices contained in the bags of $\decomptree_{\bag,i}$. Let $G_{\bag,i} := G[V_{\bag, i}]$. 
    \item For $\bag \in \bags$ and $x: \bag \cap F \rightarrow \nonneginf$ define 
    \begin{align}
        w_{\bag,i,x} : V_{\bag,i} \cap F   & \rightarrow \nonneginf \\
                      u & \mapsto \begin{cases}
                          x(u) & \textnormal{ if $u \in \bag$} \\
                          w(u) & \textnormal{ else. }
                      \end{cases}
    \end{align}
    In words: we take $w$, restrict it to $V_{\bag, i}$ and overwrite it for $v \in \bag$ by its value under $x$. 
\end{enumerate}
If $C \subseteq \bag \cup \bar{\bag}$ is a complete set of commitments, i.e., containing a commitment for each $u \in \bag$, then $(G_{\bag,i}, w_{\bag,i,x}, C)$ is an instance of Stackelberg VC with commitments on the subgraph  $G_{\bag,i}$. Our DP algorithm will compute optimal values for those instances, for each bag $\bag$, each number $0 \leq i \leq |\children(\bag)|$, each way $x$ to set the weights within $B$, and each complete set of commitment $C$. It builds a table $\rev$ in which cell $\rev[\bag,x,i,C]$ contains the value $\opt(G_{\bag,i}, w_{\bag,i,x}, C)$. 

\begin{figure}
    \centering
    \includegraphics[width=0.8\linewidth]{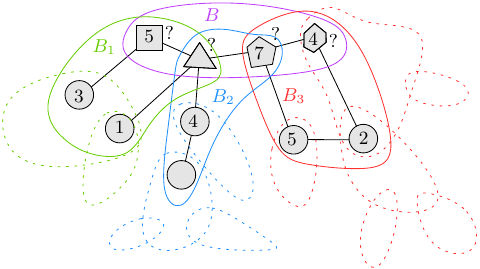}\\
    \includegraphics[width=1\linewidth]{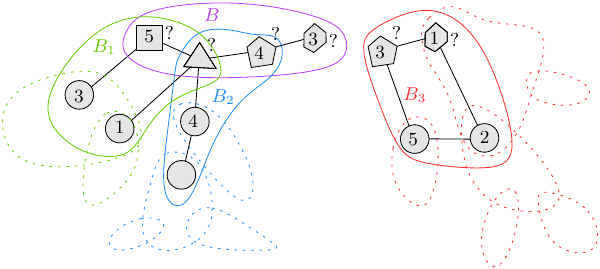}
    \caption{The instance $(G_{B,3}, w_{B,3,x}, C)$.
    The bag $\bag$ is fully committed, so we can split into two instances. We split $x$ into $y+z$ with  
    $y(\mathimg{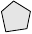})=4$,
    $y(\mathimg{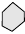})=3$,
    $z(\mathimg{img/treedecomp-03-pentagon.pdf})=3$, and 
    $z(\mathimg{img/treedecomp-04-hexagon.pdf})=1$. For the remaining vertices, $y$  and $z$ simply inherit their values from $w_{B,3,x}$. There are $(7+1)\cdot (4+1)$ ways to split $[\mathimg{img/treedecomp-03-pentagon.pdf} \mapsto 7, \mathimg{img/treedecomp-04-hexagon.pdf} \mapsto 4]$ into non-negative integer functions, so our algorithm would have to try them all.
     After the split we obtain two instances. To the left, we get $(G_{\bag,2}, w_{\bag,2,y}, C)$; the root of $G_{\bag,2}$ is fully committed so we can look up its optimal value in the DP table. To the right, we have the red instance $(G_{\bag_3, 3}, w_{\bag_3, 3, y}, C)$. Not all vertices of its root bag $\bag_3$ are committed. Our algorithm searches all four ways to commit to the vertices 
    \inlineimg{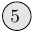} and \inlineimg{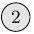}. Each way gives a complete commitment for $\bag_3$, and the optimal value can be looked up in the DP table.
    } 
    \label{fig:treedecompsplit}
\end{figure}

\textbf{Base: $i=0$.} We have to compute $\rev[\bag, x, 0, C]$. The problem instance contains only vertices in $\bag$, is fully committed, and can be solved with the algorithm from \cref{lem:compute-price-flow}.

\textbf{Step: $i \geq 1$.} We want to compute $\rev[\bag, x, i,C]$. For this let $\bag_i$ be the $i$-th child $\bag$ and write $d := |\children(\bag_i)|$. The set $\bag \cap \bag_i$ is a separator of $G_{\bag,i}$, which splits into $G_{\bag,i-1}$ and $G_{\bag_i, d}$. To realize this split, we have to divide the weights of each $u \in \bag \cap \bag_i$. That is, we have to come up with $y : \bag \cap F \rightarrow \nonneginf$ and $z : \bag_i \cap F \rightarrow \nonneginf$ such that 
\begin{align*}
    y(u) & = x(u) \qquad \textnormal{ for all $u \in \bag \setminus \bag_i$}\\
    z(u) & = w(u) \qquad \textnormal{ for all $u \in \bag_i \setminus \bag$} \\
    y(u)+z(u) & = x(u) \qquad \textnormal{ for all $u \in \bag \cap \bag_i$.} 
\end{align*}
Let ${\rm YZ}$ be the set of all pairs $(y,z)$ of functions satisfying these conditions. 
The Split-Join lemma tells us that adding the optima of the split instances never exceeds the optimum of the original one, and equals it if we guess the correct functions $y, z$. Formally:
\begin{align}
    \opt(G_{\bag,i}, w_{\bag,i,x}, C) & = \max_{(y,z) \in {\rm YZ}}  \opt(G_{\bag, i-1}, w_{\bag, i, y}, C) + 
    \opt(G_{\bag_i}, d, w_{\bag_i, d, z}, C) \ . 
    \label{eqn-max-over-x-y-z}
\end{align}
Since all weights are integers, the Split-Join lemma also tells us that it is enough to maximize over all integer-valued $y$ and $z$.
The first $\opt$ can be looked up in $\rev[B,y,i-1,C]$. The second isn't yet stored in $\rev$ since $\bag_i$ is not yet fully committed. Thus, we have to maximize over all complete set of commitments $D$ to $\bag_i$ that are {\em compatible} with $C$ (i.e., if $u \in C$ then $D$ cannot contain $\bar{u}$). Let $\mathcal{D}$ be the set of all complete commitments to $\bag_i$ that are compatible with $C$. Then 
\begin{align}
    \opt(G_{\bag_i}, d, w_{\bag_i, d, z}, C) & = 
    \max_{D \in \mathcal{D}} \opt(G_{\bag_i}, d, w_{\bag_i, d,z}, C \cup D) \\
    & = \max_{D \in \mathcal{D}} \opt(G_{\bag_i}, d, w_{\bag_i, d,z}, D) 
    \label{eqn-max-over-commitments}
\end{align}
since commitments to vertices in $\bag \setminus \bag_i$ are irrelevant for this sub-instance. We can now look up the value $\opt(G_{\bag_i}, d, w_{\bag_i, d,z}, D)$ in the DP table under  $\rev[\bag_i, z, d, D]$.

For the time complexity let $k$ be the tree width and $W = 1 + \max_{u \in F}w(u)$ the maximum weight plus one occurring in $G$. How many cells does our  $\rev[\bag,i,x,C]$ have? 
\begin{enumerate}
\item The number of possible pairs $(\bag,i)$ is $\sum_{\bag \in \bags} |\children(\bag)| = |\bags| + |\decompedges| \leq 2n-3$.
\item There are at most $W^{k+1}$ possible weight functions $x: \bag \cap F \rightarrow \{0,\dots, W-1\}$.  
\item There are $2^{|\bag|} \leq 2^{k+1}$ total commitments on $\bag$. 
\end{enumerate}
Thus, our table has $O\left(n 2^k W^{k+1}\right)$ entries. For each entry with $i \geq 1$ we execute the maximizationt steps (\ref{eqn-max-over-x-y-z}) and (\ref{eqn-max-over-commitments}). The former is over all $W^{|\bag \cap \bag_i|}$ choices of splitting $x$ as $y+z$. Actually realizing a concrete split takes itself $O(k)$ steps (subtractions). The latter maximization is over the $2^{|\bag_i \setminus \bag|}$ ways to extend the commitment to become complete on $\bag_i$. Thus the maximization takes 
\begin{align*}
    k \cdot W^{|\bag \cap \bag_i|} \cdot 2^{|\bag_i \setminus \bag|} & \leq k \cdot W^{|\bag \cap \bag_i|} \cdot W^{|\bag_i \setminus \bag|-1} \cdot 2 = 2 k  W^{|\bag_i|-1} \leq 2 k  W^{k} 
\end{align*}
steps. For $i=0$, to compute an entry $\rev[\bag, 0, x, C]$ we observe that the instance is fully committed and contains at most $k+1$ vertices. We call the Algorithm from Theorem~\ref{theorem-sell-all}, which runs in time $\poly(k)$. The total running time is thus 
\begin{align*}
    O \left( n 2^k W^{k+1} \left(k W^{k} + \poly(k) \right) \right) = O\left( nk2^k W^{2k+1}\right) \ . 
\end{align*}

\TreewidthResult*

This algorithm works only for bipartite graphs given a tree decomposition. It would be interesting to extend this to general graphs given a tree decomposition. The main obstacle to this approach is that there is no analogue of \cref{lem:split-join-sep} for general graphs. The graph $G$ in \cref{fig:counterexampleSplitJoin} has the separator $\{\pi_1, \pi_2, \pi_3\}$ that splits the graph into the outer cycle $G_1$ and the inner triangle with the priceable vertices attached ($G_2$). 

The optimal prices for $G$ sell all three vertices for the price of 1.5 each, achieving a revenue of 4.5. Selling all three priceable vertices in $G_1$ achieves a revenue of 3 by assigning prices of $(1, 1, 1)$, or $(0,1,2)$, or any permutation of the latter, or any convex combination thereof. In $G_2$ however their is no revenue to be obtained when selling all three vertices and thus we cannot write the cost function of $G$ as a combination of costs functions for $G_1$ and $G_2$ that agree on the vertex cover. Additionally, we see that for general graphs the optimal prices don't have to be integers. In \cref{sec:verification}, we provide a python script to verify and compute the optimal prices for this example.\footnote{Alternatively, you can get it at \url{https://github.com/jt0202/StackelbergVC/blob/master/bruteForceSolve.py}}

\section{Solving completely committed instances}\label{sec:computingPrices}
\label{section-completely-committed}
In the course of the algorithm, we will often have to find optimal prices given commitments for all priceable vertices. We can immediately transform this into an instance where every priceable vertex has a yes-commitments: make each priceable vertex with a no-commitment a fixed price vertex of weight $\infty$.  Given $w : F \rightarrow \nonneginf$ and a subset $Q \subseteq P$ we let $w_Q$ be the total price function obtained by setting everything in $Q$ to $0$ and everything in $P\setminus Q$ to $\infty$:
\begin{align*}
    w_Q: V & \rightarrow \R_0^+ \cup \{\infty\} \\
     u \mapsto  & 
    \begin{cases}
    0 & \textnormal{ if $u \in Q$,}\\
    \infty & \textnormal{ if $u \in P \setminus Q$,}\\
    w(u) & \textnormal{ if $u \in F$.}   
    \end{cases}
\end{align*}

Given a set $Q \subseteq P$ we write $p(Q) := \sum_{u \in Q} p(u)$.

\subsection{Sell-all Stackelberg as a linear program}

Given a (not necessarily bipartite) graph $G$ with priceable vertices $P$, we can write the sell-all optimization problem $\opt(G,w,P)$ as a linear program with $2^{|P|}$ many constraints: by definition, a pricing scheme $p : P \rightarrow \nonneg$ is feasible if and only if $(G, w+p)$ has a minimum weight vertex cover containing all of $P$. For a given set $Q \subseteq P$, what is the minimum price of a vertex cover $X$ with $X \cap P = Q$? You have to buy $Q$, have to avoid $P \setminus Q$, and are free on $F$; thus it is $p(Q) + \vc(G, w_Q)$. The pricing scheme is now feasible if this quantity is minimized by $Q = P$. We therefore define the following linear program:
\begin{definition}[Linear program computing $\opt(G,w,P)$]
We maximize over all functions $p:  P \rightarrow \nonneg$:
    \begin{align*}
        \textnormal{maximize } & p(P) \\
        \textnormal{subject to } & 
        p(P) + \vc(G, w_P) \le 
        p(Q) + \vc(G, w_Q)\quad \forall Q \subseteq P \\
         & p(u) \geq 0 \quad \forall u \in P \ . 
    \end{align*}    
\end{definition}

\begin{observation}
    The optimal value of the above linear program is $\opt(G,w,P)$.
\end{observation}

This is a specialization of a program given in \cite{DBLP:journals/algorithmica/BriestHK12}. It has $2^{|P|}$ constraints and can be defined for every Stackelberg optimization problem. For minimum weight vertex cover on bipartite graphs we can efficiently compute $\vc(G, w_Q)$ and thus efficiently check whether a solution $p$ violates a constraint; that is, we can implement as separation oracle. With the help of the ellipsoid algorithm, an optimal feasible solution to this LP can be found in polynomial time, albeit not strongly polynomial. In order to have a strongly polynomial algorithm, we consider other ways to solve the problem.

\begin{lemma}\label[lemma]{lem:opt-time}
    $\opt(G,w,P)$ can be computed in strongly polynomial time if $G$ is a bipartite and the number of priceable vertices is constant.
\end{lemma}
\begin{proof}
    Let $k = |P|$. Then we have $2^k$ many inequalities for which we compute a minimum weight vertex cover in polynomial time. We can now transform the LP into standard form with $k+2^k$ many variables and $2^k$ equations. 
    In time $\binom{2^k+k}{2^k}$ we enumerate all possible bases and in strongly polynomial time solve the linear equation system~\cite{edmonds1967systems}.
\end{proof}

The above algorithm's runtime grows exponentially with the number of priceable vertices, so that it is only feasible for small instances. Note that for at most two priceable vertices we however gain a linear time algorithm, since we do not have to invert the matrix and can instead solve it with a constant number of additions, subtractions and $\min$ operations due to the geometry of the feasible area. 

\subsection{A flow-based algorithm}

When all priceable vertices are on the same side, there is a polynomial time algorithm described in \cite{DBLP:journals/algorithmica/BriestHK12}.
We will show how to reduce our problem to the $P$-on-same-side case. First we consider $w_P$, i.e., we give all priceable vertices a price of $0$, and compute a minimum weight vertex cover $X$ for $(G,w_P)$. Without loss of generality $P \subseteq X$. Imagine an optimal feasible pricing scheme $p$ for $(G,w,P)$ (remember we have committed to selling all priceable vertices) and let $X'$ be the corresponding minimum vertex cover with respect to $w + p$. By optimality we have  
\begin{align*}
   w(X' \cap F) + p(X' \cap P) = (w + p)(X') \leq (w+p)(X) = w(X \cap F) + p(X \cap P)
\end{align*}
but since $X \cap P = X' \cap P = P$ the two terms cancel and we see that $w(X' \cap F) \leq w(X \cap F) = w_P(X \cap F)$. Since $X$ was optimal for $w_P$ this has to hold with equality. So $X$ is also optimal for $w + p$. Thus, our problem reduces to find the optimal prices $p$ such that $X$ is a minimum weight vertex cover for $G$ under $p+w$. In other words,  $\opt(G,w,P) = \opt(G,w,X)$. Let $E_X$ be the set of edges of $G$ with both endpoints in $X$. We create the subgraph $G' = (V,E\setminus E_X)$. The set $X$ is obviously still a vertex cover of $G'$, and $G'$ is still bipartite. Indeed, $V = X \dot{\cup} (V \setminus X)$ is a bipartition of $G'$ (but maybe not of $G$). With respect to this bipartition, all priceable vertices are on the same side of $G'$ and thus we can solve this instance by the algorithm of \cite{DBLP:journals/algorithmica/BriestHK12}. Correctness follows from this lemma:

\begin{lemma}
    A pricing scheme $p$ is feasible for $(G,w,X)$ if and only if it is feasible for $(G',w,X)$; consequently $\opt(G,w,X) = \opt(G',w,X)$.
\end{lemma}
\begin{proof}
    $\Leftarrow$: If $p$ is feasible for $(G',w,X)$ then $X$ is a minimum weight vertex cover for $G'$ and weight function $w + p$. By strong LP duality and integrality of the vertex cover linear program on bipartite graphs, there is a dual solution $y': E(G') \rightarrow \R$ of the same total weight. Extend $y'$ to $y : E(G) \rightarrow \R$ by setting $y(e) := 0$ for every deleted edges $e \in E \setminus E'$. This shows that $X$ is still a minimum weight vertex cover for $G$.

    $\Rightarrow$: If $p$ is a feasible pricing scheme for $(G, w, X)$, then $X$ is a minimum weight vertex cover with respect to $p$. Let $y : E \rightarrow \R$ be an optimal dual solution. By the linear complementary slackness theorem we know that $y(e) = 0$ whenever both endpoints of $e$ are in $X$. Thus, the restriction $y' := y|_{E'}$ is a dual solution for $G'$, has the same weight as $y$, and thus certifies that $X$ is also a minimum weight vertex cover for $G'$.
\end{proof}

Now all priceable vertices are on the same side and we can apply the algorithm of \cite{DBLP:journals/algorithmica/BriestHK12}. We can implement this algorithm in a different way to speed up the runtime. Let $X$ be the optimal vertex cover for $(G, w_P)$ computed above, $G'$ be the subgraph obtained from $G$ by deleting doubly-covered edges, and imagine $X$ constituting the left side of the bipartite graph $G'$.

\begin{lemma}\label{lem:compute-price-flow}
    We can compute an optimal feasible pricing scheme for $(G', w, X)$ with two calls to a  max flow algorithm.
\end{lemma}
\begin{proof}
    As in the algorithm of \cite{DBLP:journals/algorithmica/BriestHK12}, we build a network on vertex set $V \cup \{s, t\}$. In general, let $c: V \rightarrow \R$ be a {\em total price function}. For each $u \in X$ we introduce an edge $(s,u)$ and set its capacity to $\capacity(s,u) := c(u)$;  for $v \in V\setminus X$ we create an edge $(v,t)$ and set $\capacity(v,t) := c(v)$. For $\{u, v\} \in E$ with $u \in X$ we create an edge $(u,v)$ and set $\capacity(u,v) := \infty$.  Call this network $N(c)$. For a maximum flow $f$ in $N(c)$ let $S$ be the set of nodes that are reachable in the residual network $(N(c), f)$. The max-flow min-cut theorem states that this is a minimum cut; in fact, 
    \begin{align*}
        {\rm val}(f) = \capacity(S) = \vc(G, c) \ , 
    \end{align*}
    where $\capacity(S)$ is the sum of capacities of edges leaving $S$. With this notation, we build a network $N(w_P)$, i.e., imagine all $u \in P$ having price $0$ (recall the definition of $w_Q$ above) as shown in Figure~\ref{fig:network-0}.
    \begin{figure}
    \centering
    \includegraphics[width=0.7\linewidth]{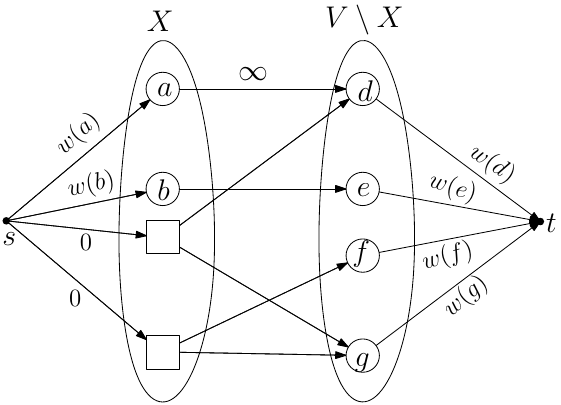}
    \caption{The flow network $N(w_P)$. Capacity into priceable vertices is $0$.}
    \label{fig:network-0}
    \end{figure}
    We compute a maximum flow $f_1$ in $N(w_P)$. Its value is 
    \begin{align*}
    {\rm val}(f_1) = \textnormal{mincut}(N(w_P)) = \vc(G', w_P) = w_P(X) \ . 
    \end{align*}
    The last equality follows because we have already computed $X$ as a minimum cut of $(G, w_P)$. Note that ${\rm val}(f_1) = w_P(X)$ immediately implies that all edges leaving $s$ are full: $f(s,v) = w_P(v)$ for all $v \in X$. We create a second network $N(w_\emptyset)$, i.e., we set $\capacity(s,u) = \infty$ for all $u \in P$. Obviously, $f_1$ is a flow in $N(w_\emptyset)$, too. We compute a maximum flow in the residual network $(N(w_\emptyset), f_1)$. Without loss of generality, this flow does not use any of the back edges $(v,s)$. The sum of those two flows is a maximum flow $f_2$ in $N(w_\emptyset)$. What is the corresponding minimum cut? Let $S_2$ be the set of vertices reachable from $s$ in the residual network $(N(w_\emptyset), f_2)$. This gives us a minimum vertex cover $X_2 := (S_2 \setminus X) \cup (S_2 \cap (V \setminus X))$: reachable vertices on the right, unreachable vertices on the left. Since $\capacity(s,u) = \infty$ for all $u \in P$, they are all reachable: $u \in S_2$, and therefore $u \not \in X_2$. So $X_2 \subseteq V \setminus X$ is a minimum weight vertex cover for $(G', w_{\emptyset})$, i.e., when we set all prices of $u \in P$ to $\infty$.

     We set $p(u) := f_2(s,u)$ for each $u \in P$ and consider a third network, $N(w + p)$. The flow $f_2$ is still a flow in $N(w+p)$, and all edges leaving $s$ are saturated. So $f_2$ is a maximum flow and $\{s\}$ is a minimum cut. Since no vertex in $X$ is reachable 
     in the residual network, $X$ is a minimum weight vertex cover for $(G', w+p)$. Since $f_2$ is a maximum 
     flow in $N(c_\emptyset)$, too, the minimum vertex cover $X_2$ harvested from $N(c_\emptyset)$ has the same weight: 
     \begin{align*}
         (w+p)(X) = {\rm val}(f_2) = w(X_2) 
     \end{align*}
     The vertex cover $X$ on left-hand side contains all of $P$; the $X_2$ on the right-hand side contains none; since any feasible pricing scheme $q$ must satisfy $(w+q)(X) \leq w(X_2)$, we see that $q(X) \leq w(X_2) - w(X)$ and thus $p$ is an optimal pricing scheme. 
\end{proof}

Therefore, we only need three max flow calls: one to compute the minimum vertex cover $X$ to transform $G$ into $G'$; the second on $N(w_P)$; the third on the residual network $N(w_{\emptyset},f_1)$. Using the result of \cite{DBLP:conf/stoc/Orlin13}, this algorithm runs in $\mathcal{O}(nm)$.

\subsection{A greedy algorithm for all priceable vertices on same side}

The above algorithm also works for trees, but is not guaranteed to run in linear time assuming that the number of priceable vertices is constant. We will show an alternative algorithm based on the linear program.
Consider the inequality of the LP associated with the set $Q \subseteq P$ and let $N := P \setminus Q$. We define the {\em penalty} $\pen(N)$ of $N$ as 
\begin{align*}
    \pen(N) := \vc(G, w_Q) - \vc(G, w_P) \ . 
\end{align*}

In words: set all prices of $u \in P$ to $0$ and take the minimum vertex cover; then increase all of $N$ to $\infty$; the price of a minimum vertex cover increases (or at least does not decrease); the penalty of $N$ then is the amount of this increase. Furthermore, we write $p(N) := \sum_{u \in N}p(u)$. We can re-write the linear program as 
\begin{align*}
\textnormal{maximize }  & p(P) \\ 
\textnormal{ subject to } & p(N) \leq \pen(N) \quad \forall N \subseteq P 
\end{align*} 

The function $\pen$ is obviously monotone, because increasing the weight of an vertex can never decrease the weight of a minimum weight vertex cover. We will additionally show that it is sub-modular when all priceable vertices are on the same side. Therefore the linear program optimizes over a polymatroid and we can find the optimum via a greedy algorithm\cite[Chapter 44]{schrijver2003combinatorial} by setting
\[p_i = \pen(\{1, \dots, i\}) - \pen(\{1, \dots, i-1\}).\]
In case $G$ is a tree, we can compute the individual penalties in linear time; the preparation step to bring all priceable vertices to the same side can also be done in linear time. 

\begin{lemma}
\label{lemma-submodular} 
    Let $(G,w)$ be a bipartite graph such that all priceable vertices are on the left side of the partition. Then the penalty function is submodular.
\end{lemma}
We will prove a more general result later in Lemma~\ref{lemma-skew-submoular}. From submodularity we can conclude the following:
\begin{lemma}\label{lem:compute-price-tree}
    Let $(T,w)$ be a weighted tree. A pricing scheme that sells all priceable vertices and maximizes the revenue can be computed in $\mathcal{O}(|P|n)$.
\end{lemma}

\subsection{A greedy algorithm for priceable vertices on both sides}

In \cref{sec:computingPrices} we gave an $\mathcal{O}(|P|n)$ algorithm to find the optimal prices to sell all priceable vertices in a graph if they are on the same side of the partition. We proved this by showing that the linear program in this case operates over a polymatroid. Now we will show a greedy algorithm that works for graphs even if both sides contain priceable vertices.

Let $G=(L \cup R, E)$ be a bipartite graph with bipartition $L, R$. Let $P$ be the priceable vertices in $L$ and $Q$ be the priceable vertices in $R$. Let $F$ be the fixed-price vertices and $w : F \rightarrow \nonneginf$ be a weight function. These can be ordered in some arbitrary way, i.e. $P = \{x_1, x_2, \dots, x_k\}$ and $Q = \{y_1, y_2, \dots, y_l\}$.

\begin{theorem}
\label{theorem-greedy-two-sides}
The pricing scheme $p : P \cup Q \rightarrow \nonneg$ define by 
\begin{align*}
    &p(x_i) = \pen(\{x_1, \dots, x_i\}) - \pen (\{x_1, \dots, x_{i-1}) \\
    &p(y_j) = \pen(\{y_1, \dots, y_j\}) - \pen(\{y_1, \dots, y_{j-1}\})
\end{align*}
is the optimal feasible pricing scheme for $G$ with weight function $w$ and commitment $P \cup Q$.
\end{theorem}

\begin{proof}
We calculate 
\begin{align}
    p(P) & = \sum_{i=1}^k p(x_i) = \sum_{i=1}^k \left(\pen(\{x_1, \dots, x_i\}) - 
    \pen (\{x_1, \dots, x_{i-1})\right) \nonumber \\
    & = \pen(P) - \pen(\emptyset) = \pen(P) \ . \label{eqn-telescoping}
\end{align}
Similarly $p(Q) = pen(Q)$. Since our linear program contains the two constraints $p(P) \leq \pen(P)$ and $p(Q) \leq \pen(Q)$, this shows that the pricing scheme is optimal. It remains to show that $p$ is feasible.

To show feasibility, we have to show that 
\begin{align}
\label{eqn-A-B-feasible}
 p(A \cup B) \leq \pen(A \cup B)   
\end{align}
holds for every $A \subseteq P$ and $B \subseteq Q$. Central to the proof will be the concept of {\em skew submodularity}:

\begin{definition}
    A function $f: 2^{P \cup Q} \rightarrow \R$ is called skew submodular if for all sets $A, A' \subseteq P$ and $B,B' \subseteq Q$ we have 
    \begin{align*}
        f(A \cup B) + f(A' \cup B') \geq f(A \cap A' \cup B \cup B') + f(A \cup A \cup B \cap B') \,  
    \end{align*}
    where $\cap$ takes precedence over $\cup$. If $B = D = \emptyset$ this becomes traditional submodularity. 
\end{definition}

\begin{figure}
    \centering
    \includegraphics[width=\linewidth]{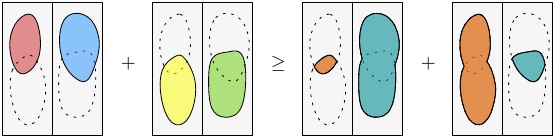}
    \caption{Skew submodularity}
    \label{fig:skew-submodularity}
    \end{figure}
\begin{lemma}
\label{lemma-skew-submoular}
    $\pen: 2^{P \cup Q} \rightarrow \R$  is skew submodular. 
\end{lemma}
  
For $A \subset P$ let $\max(A)$ denote the maximum element in the ordering $x_1 < x_2 < \dots, x_k$ of its elements, and similarly for $B \subseteq Q$. We call a set $I = \{x_1, x_2, \dots, x_i\}$ an {\em initial segment}; note that $p(I) = \pen(I)$ by a telescoping sum argument as in (\ref{eqn-telescoping}). The same holds for initial segments $J  = \{y_1, \dots, y_j\}$. 
We will show (\ref{eqn-A-B-feasible}) using induction on $\max(A)$ and $\max(B)$: if both are $0$ then both sides of (\ref{eqn-A-B-feasible}) are 0; also note that (\ref{eqn-A-B-feasible}) holds with equality whenever $A$ and $B$ are initial segments. 
  
In the induction step, we consider the initial segments  $A' = \{x_i \mid i < \max(A)\}$ and $B' = \{y_i \mid i < \max(B)\}$. Note that $A \cap A'$ has a smaller maximum than of $A$, and $A \cup A'$ is an initial segment. We apply submodularity to the two sets $A \cup B$ and $A' \cup \emptyset$:

\begin{align*}
   \pen(A \cup B) + \pen(A' \cup \emptyset) & \geq \pen(A\cup A' \cup (B \cap \emptyset)) + \pen((A \cap A') \cup B \cup \emptyset)  \\
   & = \pen(A \cup A') + \pen((A \cap A') \cup B) \\
   & = p(A \cup A') + \pen((A \cap A') \cup B) 
\end{align*}
because $A \cup A'$ is an initial segment.
We apply submodularity to $(A \cap A') \cup B$ and $\emptyset \cup B'$:
\begin{align*}
    & \pen((A \cap A') \cup B) + \pen(\emptyset \cup B') \\
    \geq \ & \pen((A \cap A' \cup \emptyset) \cup (B \cap B')) + 
    \pen((A \cap A' \cap \emptyset) \cup (B \cup B')) \\
    = \ &  
    \pen((A \cap A') \cup (B \cap B')) +  \pen(B \cup B')  \\
    = \ & \pen((A \cap A') \cup (B \cap B')) +  p(B \cup B')
\end{align*}
because $B  \cup B'$ is an initial segment. Combining the two we get 
\begin{align*}
    \pen(A \cup B) + \pen(A') + \pen(B') 
    &\ge p(A\cup A') + p(B \cup B') + \pen((A \cap A') \cup (B \cap B'))  \\
    & \geq p(A \cup A') + p(B \cup B') + p((A \cap A') \cup (B \cap B')) \tag{induction on $\max$} \\
    & = p(A) + p(A') + p(B) + p(B') \ . 
\end{align*}    
Since $A'$ and $B'$ are initial segments, $p$ and $\pen$ agree and cancel from the equation, and we are left with 
\begin{align*}
    \pen(A \cup B) \geq p(A) + p(B) = p(A \cup B) \ , 
\end{align*}
which shows (\ref{eqn-A-B-feasible}).
\end{proof}

\begin{proof}[Proof of Lemma~\ref{lemma-skew-submoular}]
    By definition, $\pen(N) = \vc(G, w_{(P \cup Q) \setminus N} - \vc(G, w_{P \cup Q})$. The second term is the same for all $N$, so it is equivalent to show that $N \mapsto \vc(G, w_{(P \cup Q) \setminus N}$ is skew submodular. This is the cost of a minimum vertex cover that avoids $N$ when vertices of $(P \cup Q) \setminus N$ have cost zero. 

    To ease notation, let $c: \L \cup R \rightarrow \R$ be an arbitrary cost function. For $A \subseteq P$ and $B \subseteq Q$ let $\vca(A,B)$ be the cost of a minimum weight vertex cover in $(L \setminus A) \cup (R \setminus B)$; vca is an acronym of {\em vertex cover avoiding} $A \cup B$. Note that $\vc(G, w_{(P \cup Q) \setminus N} = \vca(N \cap L, N \cap R)$. We show that $\vca$ is skew submodular, namely 
    \begin{align*}
        \vca(A,B) + \vca(A,B') \geq \vca(A \cap A', B \cup B') + \vca(A \cup A', B \cap B') \ . 
    \end{align*}
    Let's start with the left-hand side and pick vertex covers $X \cup Y$ with $X\subseteq L\setminus A, Y \subseteq R \setminus B$ and $X'\subseteq L\setminus A', Y' \subseteq R \setminus B'$ of weight $\vca(A,B)$ and $\vca(A',B')$, respectively. Here are two operation that combine two vertex covers into a new one:
    \begin{align*}
        {\rm left} & := (X \cup X') \cup (Y \cap Y') \\ 
        {\rm right} & := (X \cap X') \cup (Y \cup Y') \ . 
    \end{align*}
    To see that ${\rm left}$ is a vertex cover, consider $e = \{u,v\}$. If $v \in Y \cap Y'$ then it is obviously covered by left; otherwise, without loss of generality $v \not \in Y'$; then $u \in X'$ since $X' \cup Y'$ is a vertex cover. Thus $u \in {\rm left}$ and the edge is covered. The argument for right is symmetrical. Note that 
    \begin{align*}
    \vca(A,B) + \vca(A',B') & = c(X \cup Y) + c(X' \cup Y') = c(X) + c(Y) + c(X') + c(Y') \\
        & = c(X \cup X') + c(X \cap X') + c(Y \cup Y') + c(Y \cap Y') \\
        & = c(X \cup X' \cup (Y \cap Y')) + c((X \cap X') \cup Y \cup Y')  \ . 
    \end{align*}
    Now $X \cup X' \subseteq (L \setminus A) \cup (L \setminus A') = L \setminus (A \cap ')$, and $Y \cap Y' \subseteq (R \setminus B) \cap (R \setminus B') = R \setminus (B\cup B')$. So the vertex cover $X \cup X' \cup (Y \cap Y')$ avoids $A \cap A'$ and $B \cup B'$ and thus its cost is at least $\vca(A \cap A', B \cup B')$. A similar argument shows that $c((X \cap X') \cup Y \cup Y') \geq \vca(A \cup A', B \cap B')$. Therefore 
    \begin{align*}
       c(X \cup X' \cup(Y \cap Y')) + c((X \cap X') \cup Y \cup Y') & \geq 
       \vca(A \cap A', B \cup B') + \vca(A \cup A', B \cap B') . 
    \end{align*}
    Putting everything together shows that $\vca$ is skew submodular, and thus is $\pen$, which is simply $\vca$ minus a constant.
\end{proof}

\subsection{Optimal prices are integers}

We can now give a full proof of \cref{lem:int-prices}.

\begin{lemma}
    Any Stackelberg vertex cover instance on a bipartite graph with integer weights for the fixed price vertices has an optimal solution where every priceable vertex receives an integer price (or infinity).
\end{lemma}
\begin{proof}
    Here is a not very efficient algorithm for computing an optimal pricing scheme:
    \begin{enumerate}
        \item Iterate over all commitments, i.e., sets $X \cup \bar{Y} \in P \cup \bar{P}$ such that $P = X \dot{\cup} Y$.
        \item Convert all $v \in Y$  to fixed-price vertices with $w(v) = \infty$.
        \item Solve the sell-all instance $(G, w, X)$, for example using the greedy algorithm described before. 
        \item Return the pricing scheme yielding the maximum revenue.
    \end{enumerate}
    All numbers appearing in the greedy algorithm are weights of minimum vertex covers for specific integral (plus infinity) weight functions, and sums and differences thereof. Thus, this algorithm produces a pricing scheme in which all prices are integers or infinity.
\end{proof}

\section{Solving cycles}

Let $C$ be a cycle and contain at least one priceable vertex $p$. Either we sell all vertices in the optimal solution or there is at least one vertex that is not bought. The former can be computed in polynomial time, since the decision problem is in \P~\cite{DBLP:conf/icalp/BohnleinKS17}. In the latter case, this transforms into a path, which is an LCA-tree and thus solvable in polynomial time by \cref{thm:LCA-correct}.

A more efficient way is possible for cycles of an even length with at least one priceable vertex $p$. In the first option we again set $p$ to infinity and solve the path. The second option, we can transform it into a path by adding a copy of $p'$ that use one of the edges of $p$ instead and have a commitment to buy both $p$ and $p'$ in the optimal solution. We can solve this analogously as in \cref{sec:LCA-trees}, but exclude the options to not buy $p$ or $p'$ in the dynamic programming steps. Since this is a bipartite cycle, the dual solutions allow us to split and join $p$ and $p'$ in a solution as in \cref{lemma-join-split}
\section{NP-hardness of Stackelberg vertex cover on trees with commitment}
\label{np-hard-with-commitment}

In our approach we compute for all distributions of the weight to the subtrees the optimal solution. One might wonder whether there is a more clever way of finding the right distribution as the optimal distribution allows us to solve the problem very easily. This however seems unlikely. We show here for completeness that even one commitment suffices for \NP{}-hardness. Note that this is superseeded by \cite{NPCTree}.

\begin{theorem}
    The Stackelberg-Vertex-Cover problem with commitments is $\NP{}$-hard even with a single commitment. 
\end{theorem}
\begin{proof}
    We reduce from the partition problem. Given a set $ A = \{a_1, a_2, \dots, a_n\}$ we construct the following branch for each element $a_i \in A$ with the common root $v$ of weight $\frac{1}{2}\sum_{i=1}^n a_i$. All other weights are denoted in the picture.
    \begin{figure}
        \centering
        \begin{tikzpicture}[every node/.style={circle,draw}]
            \node(a) at (0,0) {$v$};
            \node(b) at (2, 0) {$b : a_i$};
            \node(c) at (4, 0) {$c : 2a_i$};
            \node(d) at (6, -1) {$\pi_i$};
            \node(e) at (6, 1) {$d: \frac{a_i}{2}$};
            \node(f) at (8, 1) {$\pi_i'$};
        
            \draw (a) -- (b);
            \draw (b) -- (c);
            \draw (c) -- (d);
            \draw (c) -- (e);
            \draw (e) -- (f);
        \end{tikzpicture}
        \caption{A branch in the construction with the common root $v$}
    \end{figure}
        
    We claim that it is possible to get a vertex cover that includes $v$ and gets a revenue of $\frac{3}{4}\sum_{i=1}^n a_i$ iff the partition instance has a solution. Using the split-join lemma, this is equivalent to finding a distribution of the weight of $v$ to the individual branches, such that $v$ is sold in every branch and that maximizes the revenue.

    We can never distribute more than $a_i$ to the $a_i$-branch or else $b$ would be sold instead of $v$ in this branch. Additionally, not both $\pi_i$ and $\pi_i'$ can be sold for a positive amount, since at least one of $c$ or $d$ must be in the vertex cover to cover the edge $\{c,d\}$. Therefore, one of the priceable vertices $\pi_i$ and $\pi_i'$ must have a neighbor $x$ in the vertex cover and the neighbor of $x$ can be removed from the vertex cover if it has a positive price decreasing the cost.

    If $\pi_i$ is not sold, then we can set the price to infinity and set $\pi_i'$ to its maximum weight of $\frac{a_i}{2}$. If $\pi_i$ is sold instead, then it is better to set $\pi_i'$ to infinity so that the weight of $d$ does not decrease the possible price for $\pi_i$. Additionally, we see that $\pi_i$ can only be sold for a positive amount in a vertex cover that includes $v$, if $v$ has in this branch a weight of $0$ and if $\pi_i$ is set to more than $a_i$ it is better to sell $c$ instead of $b$ and $\pi_i$ 

    Using this, we can achieve $\frac{3}{4}\sum_{i=1}^n a_i$ when we have a solution $Y$ of the partition instance by selling $\pi_i$ if $a_i \in Y$ and $\pi_i'$ else.

    On the other hand, if we achieve a revenue of $\frac{3}{4}\sum_{i=1}^n a_i$ with a vertex cover $X$ that includes $v$, then there must a set $X'$ of $i$ such that $\pi_i$ is sold for $a_i$. The set $X'$ is a solution for the partition instance, because $Y = \{1, \dots, n\} \setminus X'$ is the set of branches which distributed accepted the weight and thus $\sum_{i \in Y} a_i = \frac{1}{2} \sum_{i=1}^n a_i$.
\end{proof}

Note that this construction actually has a better solution that does not sell $v$. We can set $\pi_i'$ always to infinity and buy $b$ and $\pi_i$ in every branch. It is  also not possible to directly encode the requirement of buying $v$ by combining it with a node of very large weight, because then the whole weight of $v$ can be distributed to the branch of the very large weight node.

This theorem tells us that in order to find an optimal solution in polynomial time, we must somehow see that buying $v$ here is the worse option without actually computing its value.
\section{Code for the verification of figure~\ref{fig:counterexampleSplitJoin}}\label{sec:verification}
\begin{python}
import numpy as np
from scipy.optimize import linprog, milp, LinearConstraint
from more_itertools import powerset

# Should be higher than the sum of weights of all neighbors
# for any priceable vertex
INFINITY = 1000

# Joins a list of dictionaries with values by summing values associated with the same key
# We don't use Counter because we want to keep zero elements
def joinDicts(l):
    res = {}
    for d in l:
        for k, v in d.items():
            if k in res:
                res[k] += v
            else:
                res[k] = v
    return res

# Transform the graph so that the priceable vertices (which are identified by weight 0)
# occur always at the start. This is assumed for the later algorithms to work.
# Additionally functions to transform the graph into the new format and out of it are provided.
def canonicalize(weights, edges):
    vertices = sorted(weights)
    mapping = {v: i for i, v in enumerate(vertices)}
    inv_mapping = {v : k for (k, v) in mapping.items()}
    
    canonical_edges = [[mapping[u] for u in edge] for edge in edges]
    canonical_weights = [weights[v] for v in vertices]

    return canonical_weights, canonical_edges, mapping, inv_mapping

# Find the minimum weight vertex cover for a general graph with an ILP approach.
def vertex_cover_cost(weights, edges):
    A = np.zeros((len(edges), len(weights)))
    for row,e in enumerate(edges):
        A[row, e[0]] = 1
        A[row, e[1]] = 1
    res = milp(c=np.array(weights), constraints=LinearConstraint(A, 1, np.inf), integrality=np.full_like(weights, True))
    return res.fun

# Creates the large LP to compute the optimal prices. If a specific selection is given, then
# the optimal prices when selling exactly the selection are computed.
def find_optimal_prices(weights, edges, selection=None):
    priceable_vertices = []
    for i in range(len(weights)):
        if weights[i] == 0:
            priceable_vertices.append(i)

    subsets = list(powerset(priceable_vertices))
    costs = {}

    for bought_elements in subsets:
        costs[bought_elements] = vertex_cover_cost(
            [INFINITY if i in priceable_vertices and i not in bought_elements else x for i, x in enumerate(weights)], edges)
    best_selection = ()
    best_price = []
    best_revenue = -100

    if selection is None:
        subsets.remove(())
    else:
        subsets = [selection]
    for bought_elements in subsets:
        A = []
        b = []
        for subset in powerset(bought_elements):
            b.append(costs[subset] - costs[bought_elements])
            A.append([0 if x in subset else 1 for x in bought_elements])
        res = linprog(np.array([-1]*len(bought_elements)), np.array(A), np.array(b))

        if -res.fun >= best_revenue:
            best_selection = bought_elements
            best_price = res.x
            best_revenue = - res.fun

    return (best_selection, best_price, best_revenue)

def solve(weights, edges, selection=None):
    weights, edges, mapping, inverse_mapping = canonicalize(weights, edges)

    if selection is not None:
        selection = tuple(mapping[v] for v in selection)

    selection, prices, revenue = find_optimal_prices(
        weights, edges, selection
    )

    selection = tuple(inverse_mapping[v] for v in selection)

    return selection, prices, revenue

G1_weights = {
    0: 0,
    1: 1,
    2: 0,
    3: 1,
    4: 0,
    5: 1,
}

G1_edges = [
    (0, 1),
    (1, 2),
    (2, 3),
    (3, 4),
    (4, 5),
    (5, 0),
]

G2_weights = {
    0: 0,
    2: 0,
    4: 0,
    6: 2,
    7: 2,
    8: 2,
}

G2_edges = [
    (0, 6),
    (2, 8),
    (4, 7),
    (6, 7),
    (7, 8),
    (6, 8),
]

combined_weights = joinDicts([G1_weights, G2_weights])
combined_edges = G1_edges + G2_edges

selection, prices, revenue = solve(combined_weights, combined_edges)

print("Combined:", selection, prices, revenue)
print("G1:", solve(G1_weights, G1_edges, selection))
print("G2:", solve(G2_weights, G2_edges, selection))
\end{python}
\section{Example on the path}

\begin{figure}[H]
    \centering
    \includegraphics[width=\textwidth]{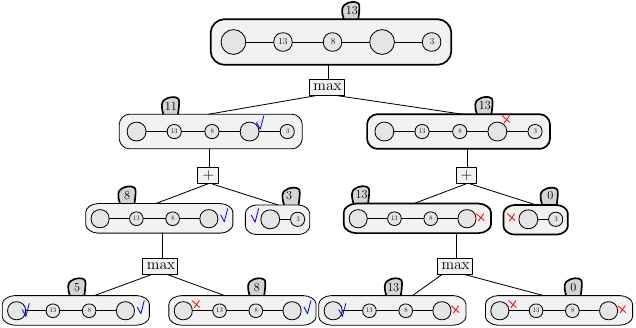}
    \caption{Our algorithm working on a path. The example is a reduced version of the one 
used in~\cite{DBLP:conf/sagt/EickhoffKP23}, and we present it in a recursive fashion rather than 
in its dynamic programming version. Blue checkmarks mean a yes-commitment; red crosses a no-commitment.
For each leaf, the value can be computed by a linear program with up to two variables and three constraints;
or by hand.}
    \label{fig-path-example}
\end{figure}

\end{document}